\documentclass[11pt,letterpaper]{article}

\usepackage{geometry}
\usepackage{thmtools,thm-restate}
\usepackage{amsthm,amsmath,amssymb}
\usepackage{cite}
\usepackage{bbm}
\usepackage{microtype}
\usepackage{graphicx}
\usepackage{enumerate}
\usepackage[dvipsnames]{xcolor}

\definecolor{nblue}{RGB}{180, 40, 90}
\definecolor{dark-gray}{gray}{0.35}
\usepackage[colorlinks=true,citecolor=nblue,linkcolor=nblue,urlcolor=blue]{hyperref}

\theoremstyle{plain}
\newtheorem{theorem}{Theorem}
\newtheorem{lemma}[theorem]{Lemma}
\newtheorem{corollary}[theorem]{Corollary}
\newtheorem{proposition}[theorem]{Proposition}

\theoremstyle{definition}
\newtheorem{definition}[theorem]{Definition}

\theoremstyle{remark}
\newtheorem{remark}[theorem]{Remark}

\usepackage{enumitem}

\usepackage{tikz}
\usetikzlibrary{arrows,automata,positioning}
\usepackage[latin1]{inputenc}
\usepackage{pgfplots}
\pgfplotsset{compat=1.13}

\usepackage{newtxtext}
\usepackage{newtxmath}

\usepackage{comment}
\usepackage{relsize}

\usepackage{bm}
\usepackage{mdframed}
\global\mdfdefinestyle{exampledefault}{%
	linecolor=gray,linewidth=0.6pt,%
	leftmargin=2pt,rightmargin=2pt
}

\newcommand{\myd}{\bm{d}}

\newcommand{\N}{\mathbb{N}}

\newcommand{\R}{\mathbb{R}}

\newcommand{\Z}{\mathbb{Z}}

\newcommand{\Gd}{\mathcal{G}(\myd)}
\newcommand{\Gul}{\mathcal{G}(\bm{\ell},\bm{u})}
\newcommand{\Gmul}{\mathcal{G}_m(\bm{\ell},\bm{u})}
\newcommand{\Dm}{\mathcal{D}_m(\bm{\ell},\bm{u})}

\makeatletter
\g@addto@macro\bfseries{\boldmath}
\makeatother

\usepackage{ifthen}
% ...
\newenvironment{rtheorem}[3][]{%
\noindent\ifthenelse{\equal{#1}{}}{\textbf{#2} #3.}{ #2 #3 (#1)}%
\begin{it}}{\end{it}}

\newcommand{\mysetminusD}{\hbox{\tikz{\draw[line width=0.6pt,line cap=round] (3pt,0) -- (0,6pt);}}}
\newcommand{\mysetminusT}{\mysetminusD}
\newcommand{\mysetminusS}{\hbox{\tikz{\draw[line width=0.45pt,line cap=round] (2pt,0) -- (0,4pt);}}}
\newcommand{\mysetminusSS}{\hbox{\tikz{\draw[line width=0.4pt,line cap=round] (1.5pt,0) -- (0,3pt);}}}
\newcommand{\mysetminus}{\mathbin{\mathchoice{\mysetminusD}{\mysetminusT}{\mysetminusS}{\mysetminusSS}}}

\DeclareMathOperator{\gap}{Gap}

\numberwithin{theorem}{section}
\numberwithin{equation}{section}
\begin{document}
\title{
Approximate Sampling and Counting of Graphs \\ with Near-Regular Degree Intervals
}
\author{Georgios Amanatidis \\
	 University of Essex\\
	 Colchester, United Kingdom\\
	 \small{\texttt{georgios.amanatidis@essex.ac.uk}}\\
	\and
Pieter Kleer\footnote{\ Part of this work has been carried out while the author was a postdoctoral fellow at the Max Planck Institute for Informatics in Saarbr\"ucken, Germany}\\
Tilburg University\\
Tilburg, The Netherlands\\
\small{\texttt{p.s.kleer@tilburguniversity.edu}}
}

\date{\today}
\begin{titlepage}
	\clearpage
	\maketitle
	\thispagestyle{empty}

\begin{abstract}
\noindent The approximate uniform sampling of graphs with a given degree sequence is a well-known, extensively studied problem in theoretical computer science and has significant applications, e.g., in the analysis of social networks. In this work we study a generalization of the problem, where \emph{degree intervals} are specified instead of  a single degree sequence. We are interested in sampling and counting graphs whose degree sequences satisfy the corresponding degree interval constraints. A natural scenario where this problem arises is in hypothesis testing on networks that are only partially observed. 
We provide the first \emph{fully polynomial almost uniform sampler (FPAUS)} as well as the first \emph{fully polynomial randomized approximation scheme (FPRAS)} for sampling and counting, respectively, graphs with near-regular degree intervals, i.e., graphs in  which every node has a degree from an interval not too far away from a given $r \in \N$. In order to design our FPAUS, we rely on various state-of-the-art tools from Markov chain theory and combinatorics. In particular, by carefully using Markov chain decomposition and comparison arguments, we reduce part of our problem to the recent breakthrough of Anari, Liu, Oveis Gharan, and Vinzant (2019) on sampling a base of a matroid under a strongly log-concave probability distribution, and we provide the first non-trivial algorithmic application of a  breakthrough asymptotic enumeration formula of Liebenau and Wormald (2017).
As a more direct approach, we also study a natural Markov chain  recently introduced by Rechner, Strowick and M\"uller-Hannemann (2018), based on three local operations---switches, hinge flips, and additions/deletions of an edge. We obtain the first theoretical results for this Markov chain, showing it is rapidly mixing for the case of near-regular degree intervals of size at most one.
\end{abstract}

\end{titlepage}

\newpage

\clearpage
\setcounter{tocdepth}{2}
\tableofcontents
\thispagestyle{empty}

\newpage

\setcounter{page}{1}

\newpage
\setcounter{page}{1}

	\section{Introduction}
\noindent The (approximate) uniform sampling and counting of graphs with given degrees has received a lot of attention during the last few decades, see, e.g.,
\cite{AK2019,Bayati2010,BenderC1978,Blitzstein2011,Bollobas1980,CarstensK18,Coolen2017,DyerGKRS21,Cooper2007,CooperDGH17,Cooper2012corrigendum,Jerrum1990,McKay1990,Kannan1999,Steger1999,Feder2006,Kim2006,Erdos2013,Erdos2015decomposition,ErdosMMS2018,Gao2017,Greenhill2017journal,GaoW18,Erdos2022stable,Liebenau2017,McKay1990conjecture}.
Given a degree sequence $\myd = (d_1,\dots,d_n)$, the goal of approximate uniform sampling  is to design a randomized algorithm that outputs a labelled simple undirected graph $G$ with degree sequence $\myd$, according to a distribution that is close to the uniform distribution over the set of all graphs with this  degree sequence. Such an algorithm is called an \emph{approximate (uniform) sampler}. Approximate samplers find applications in fields such as complex network analysis, where they serve as null models for hypothesis testing. Consider, e.g., a social network with edges representing friendships or relationships. One might see a very high number of edges between a certain group of nodes and, based on this, conjecture that these nodes form a \emph{community} of friends or colleagues. In order to test this hypothesis, one would like to be able to generate graphs with \emph{similar characteristics} as the observed network and, based on these generated samples, decide how likely it is that there is a high number of edges between that particular group of nodes by chance alone. Here the characteristic of interest is the degree sequence of the observed network \cite{Olding2014}. For determining how many samples are sufficient in order to test the hypothesis, we also need to be able to count the number of graphs with the given degree sequence.

In practice, it is not always possible to have exact knowledge of the degree sequence of an observed network, due to erroneous measurements. In order to overcome this, there is a need for more robust null models. One such model was proposed by Rechner, Strowick and M\"uller-Hannemann  \cite{Rechner2018}. Instead of a given degree sequence $\myd$, the null models now consist of all graphs with given \emph{degree interval} constraints $[\ell_i,u_i]$, for $i \in [n] = \{ 1,\dots,n \}$. In this case we say that a graph $G$ has degrees in the interval $[\bm{\ell},\bm{u}]$ with $\bm{\ell} = (\ell_1,\dots,\ell_n)$ and $\bm{u} = (u_1,\dots,u_n)$. The algorithmic task at hand then becomes to develop algorithms for sampling and counting graphs from the set $\mathcal{G}(\bm{\ell},\bm{u})$ of all graphs satisfying the interval constraints. An intuitive two-step approach for solving this problem is to first sample \emph{according to the correct proportional distribution} a degree sequence $\myd = (d_1,\dots,d_n)$ from the set of all degree sequences satisfying the interval constraints $\ell_i \leq d_i \leq u_i$, for $i \in [n]$, and then sample uniformly at random a graphical realization from the set $\mathcal{G}(\myd)$ of all graphs with degree sequence $\myd$. A crucial difficulty that arises here is that the probability with which each degree sequence $\myd$ needs to be sampled in the first step is not obvious. This probability should be proportional to the number $|\Gd|$ which is not known in general.

To make the problem more concrete, we give a brief example in the context of the social network application that we started out with. Suppose we have a partially observed network. For a given node $i$, we let $\ell_i$ be the number of observed edges adjacent to $i$, $\delta_i$ the number of missing observations and, thus, $u_i = n - 1 - (\ell_i + \delta_i)$ the number of observed non-edges (i.e., pairs $\{i,j\}$ for which we know there is no edge between nodes $i$ and $j$).
There are now two extreme cases: either all missing observations are non-edges, meaning that node $i$ has degree $\ell_i$, or all missing observations are indeed edges, meaning that node $i$ has degree $u_i$. Hence, we are interested in  sampling (and counting)  graphs for which each node $i$ has a degree in the interval $[\ell_i,u_i]$, for every $i \in [n]$. 
In this and other similar settings, these problems seem to be natural and elegant generalizations of the classic graph sampling and counting problems.

Towards sampling graphs with given degree intervals, Rechner et al.~\cite{Rechner2018} introduced a Markov chain based on three simple operations: \emph{switches}, \emph{hinge flips} and \emph{additions/deletions}. The chain in each step selects one of these operations uniformly at random and  performs it, if possible. We call this chain the \emph{degree interval Markov chain}. The operations are shown in
Figure \ref{fig:add_del} and a formal definition is given in Section \ref{sec:preliminaries} and Appendix \ref{app:interval_chain}.
These three operations are  the ones  described by Coolen et al.~\cite{Coolen2017} as the most commonly used operations in Markov Chain Monte Carlo algorithms for the generation of simple undirected graphs \emph{in practice}.  This serves as additional motivation for rigorously studying Markov chains based on these operations. We will also be interested in the \emph{switch-hinge flip Markov chain} that only uses the switch and hinge flip operations. The hinge flip and switch operations are of particular interest both  in theory and in practice as they preserve the number of edges and the degree sequence of a graph, respectively.

	\begin{figure}[t]
		\centering
		\begin{tikzpicture}[scale=0.3]
		\coordinate (A1) at (0,0);
		\coordinate (A2) at (0,2);
		\coordinate (M1) at (2.5,0);
		\coordinate (M2) at (2.5,2);

		\coordinate (P1) at (4,1);
		\coordinate (P2) at (5.5,1);

		\node at (A1) [circle,scale=0.7,fill=black] {};
		\node (a1) [below=0.1cm of A1]  {$v$};
		\node at (A2) [circle,scale=0.7,fill=black] {};
		\node (a2) [above=0.1cm of A2]  {$w$};
		\node at (M1) [circle,scale=0.7,fill=black] {};
		\node (m1) [below=0.1cm of M1]  {$x$};
		\node at (M2) [circle,scale=0.7,fill=black] {};
		\node (m2) [above=0.1cm of M2]  {$y$};

		\path[every node/.style={sloped,anchor=south,auto=false}]
		(A1) edge[-,very thick] node {} (A2)
		(M1) edge[-,very thick] node {} (M2)
		(P1) edge[<->,very thick] node {} (P2);
		\end{tikzpicture}
		\begin{tikzpicture}[scale=0.4]
		\coordinate (A1) at (0,0);
		\coordinate (A2) at (0,1.5);
		\coordinate (M1) at (2,0);
		\coordinate (M2) at (2,1.5);


		\node at (A1) [circle,scale=0.7,fill=black] {};
		\node (a1) [below=0.1cm of A1]  {$v$};
		\node at (A2) [circle,scale=0.7,fill=black] {};
		\node (a2) [above=0.1cm of A2]  {$w$};
		\node at (M1) [circle,scale=0.7,fill=black] {};
		\node (m1) [below=0.1cm of M1]  {$x$};
		\node at (M2) [circle,scale=0.7,fill=black] {};
		\node (m2) [above=0.1cm of M2]  {$y$};

		\path[every node/.style={sloped,anchor=south,auto=false}]
		(A1) edge[-,very thick] node {} (M2)
		(M1) edge[-,very thick] node {} (A2);
		\end{tikzpicture}
		\qquad\qquad\qquad
		\begin{tikzpicture}[scale=0.3]
		\coordinate (i) at (0,0);
		\coordinate (k) at (4,0);
		\coordinate (j) at (2,2);

		\coordinate (P1) at (5,1);
		\coordinate (P2) at (6.5,1);

		\node at (i) [circle,scale=0.7,fill=black] {};
		\node (i1) [below=0.1cm of i]  {$v$};
		\node at (j) [circle,scale=0.7,fill=black] {};
		\node (j1) [above=0.1cm of j]  {$w$};
		\node at (k) [circle,scale=0.7,fill=black] {};
		\node (k1) [below=0.1cm of k]  {$x$};

		\path[every node/.style={sloped,anchor=south,auto=false}]
		(i) edge[-,very thick] node {} (j)
		(P1) edge[<->,very thick] node {} (P2);
		\end{tikzpicture}
		\begin{tikzpicture}[scale=0.3]
		\coordinate (i) at (0,0);
		\coordinate (k) at (4,0);
		\coordinate (j) at (2,2);

		\node at (i) [circle,scale=0.7,fill=black] {};
		\node (i1) [below=0.1cm of i]  {$v$};
		\node at (j) [circle,scale=0.7,fill=black] {};
		\node (j1) [above=0.1cm of j]  {$w$};
		\node at (k) [circle,scale=0.7,fill=black] {};
		\node (k1) [below=0.1cm of k]  {$x$};

		\path[every node/.style={sloped,anchor=south,auto=false}]
		(j) edge[-,very thick] node {} (k);
		\end{tikzpicture}
		\qquad\qquad\qquad
		\begin{tikzpicture}[scale=0.3]
		\coordinate (i) at (0,0);
		\coordinate (j) at (0,2);

		\coordinate (P1) at (1.5,1);
		\coordinate (P2) at (3,1);

		\node at (i) [circle,scale=0.7,fill=black] {};
		\node (i1) [below=0.1cm of i]  {$v$};
		\node at (j) [circle,scale=0.7,fill=black] {};
		\node (j1) [above=0.1cm of j]  {$w$};

		\path[every node/.style={sloped,anchor=south,auto=false}]
		(i) edge[-,very thick] node {} (j)
		(P1) edge[<->,very thick] node {} (P2);
		\end{tikzpicture}
		\begin{tikzpicture}[scale=0.3]
		\coordinate (i) at (0,0);
		\coordinate (j) at (0,2);

		\coordinate (P1) at (2,1);
		\coordinate (P2) at (3,1);

		\node at (i) [circle,scale=0.7,fill=black] {};
		\node (i1) [below=0.1cm of i]  {$v$};
		\node at (j) [circle,scale=0.7,fill=black] {};
		\node (j1) [above=0.1cm of j]  {$w$};

		\end{tikzpicture}
		\caption{Left to right: switch on $v,w,x,y$; hinge flip on $v,w,x$; edge addition/deletion on $v,w$.}
		\label{fig:add_del}
	\end{figure}
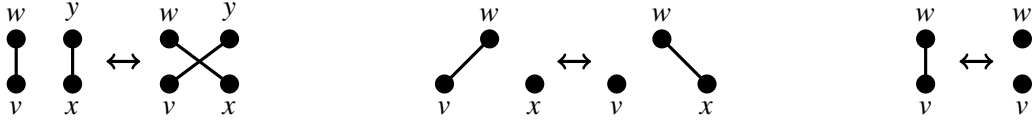

\subsection{Our contributions}
\noindent In this work, we give the first efficient approximate sampler and approximate counter for graphs with so-called \emph{near-regular} degree intervals. Near-regularity here refers to the fact that all graphs have degrees which are close to a common value up to a sublinear margin. To be more precise, we show that there is
a \emph{fully polynomial almost uniform sampler (FPAUS)} and a \emph{fully polynomial randomized approximation scheme (FPRAS)} (for formal definitions see Section \ref{sec:preliminaries}),
in case the degree intervals are close to a common value $r= r(n) \in \N$, i.e., if $[\ell_i,u_i] \subseteq [r - r^\alpha,r+r^\alpha]$ for some $0 < \alpha < \frac{1}{2}$. The parameter $\alpha > 0$ models the maximum length of the degree intervals that we allow; this length of $2 r^\alpha$ should be relatively small compared to $r$.\footnote{\ One should note that an assumption of this kind is to be expected. Otherwise, we would be also solving the problem of (approximately) uniformly sampling a graph with \emph{any} given degree sequence, which is a long-standing open problem.} We also need a minor  technical assumption on the value of $r$ in order to avoid some (arguably not very interesting) boundary cases. The main result of this work is Theorem \ref{thm:main_result} below.

For vectors $\bm{a} = (a_1,\dots,a_n), \bm{b} = (b_1,\dots,b_n) \in \R^n$, we write $\bm{a} \leq \bm{b}$ if $a_i \leq b_i$ for all $i \in [n]$.
Given $\bm{\ell}, \bm{u}\in \N^n$, by $\Gul$ we denote the set of all graphs $G$ whose degree sequence $\myd(G)$ satisfies $\bm{\ell} \leq \myd(G) \leq \bm{u}$.

\begin{theorem}\label{thm:main_result}
	Let $0 < \alpha < 1/2$ and $0 < \sigma < 1$ be fixed. Let $r = r(n)$ with $2 \leq r \leq (1-\sigma)n$. If for every node $i \in [n]$ it holds that $[\ell_i,u_i] \subseteq [r - r^\alpha,r+r^\alpha]$, then there is an FPAUS for the approximate uniform sampling of graphs from $\Gul$ and an FPRAS for approximating $|\Gul|$.
\end{theorem}

For given degree intervals $[\bm{\ell},\bm{u}]$ and $m \in \N$, we write $\Gmul$ for the set of graphs $G$ whose degree sequence $\myd(G)$ satisfies $\bm{\ell} \leq \myd(G) \leq \bm{u}$ \emph{and} $\sum_i d_i = 2m$.
By using reductions between approximate sampling and approximate counting (see Appendix \ref{sec:counting}) we get that to prove Theorem \ref{thm:main_result} it suffices to show the existence of an FPAUS for sampling from $\Gmul$. To this end, we show that the switch-hinge flip Markov chain is rapidly mixing under the conditions of Theorem \ref{thm:main_result}.
This result is summarized in Theorem \ref{thm:main_switch_hinge}.

\begin{theorem}\label{thm:main_switch_hinge} 
	Let $\alpha$, $\sigma$, and $r$ be as in Theorem \ref{thm:main_result}. If $[\ell_i,u_i] \subseteq [r - r^\alpha,r+r^\alpha]$, for all  $i \in [n]$, and  $2m\in \big[\sum_i \ell_i, \sum_i u_i \big]$, then the switch-hinge flip Markov chain is rapidly mixing on $\Gmul$.
\end{theorem}

A more direct approach for sampling from $\Gul$ than the one behind Theorem \ref{thm:main_result} would be to use the degree interval Markov chain. An interesting open question is whether this chain is rapidly mixing under the assumptions in Theorem \ref{thm:main_result} (or under weaker assumptions). As a first step into this direction, we show rapid mixing when all the degree intervals have size at most one, i.e., when $u_i - 1 \leq \ell_i \leq u_i$.
	\begin{theorem}\label{thm:main_chain}
		Let $\alpha$, $\sigma$, and $r$ be as in Theorem \ref{thm:main_result}. If $[\ell_i,u_i] \subseteq [r - r^\alpha,r+r^\alpha]$ and  $u_i - 1 \leq \ell_i \leq u_i$, for all  $i \in [n]$, then the degree interval Markov chain is rapidly mixing on $\Gul$.
	\end{theorem}
	
	The technical novelty of our work lies in the highly nontrivial combination of state-of-the-art tools from Markov chain theory and combinatorics. An overview of our proof approach is given in Section \ref{sec:main_result}. It relies on Markov chain decomposition and comparison techniques of Martin and Randall \cite{Martin2006}, rapid mixing results for the switch Markov chain by Amanatidis and Kleer \cite{AK2019}, the breakthrough work of Anari et al.~\cite{SLC2} on strongly log-concave probability distributions, and the work of Liebenau and Wormald \cite{Liebenau2017} regarding asymptotic enumeration formulas for the number of near-regular graphs.
	
	\begin{remark} 
		Our theorems---and all the building blocks used in their proofs---are shown to be true for all $n \geq n_0$, where $n_0 \in \N$ is a constant that depends on the other constant parameters involved. It is straightforward that for $n < n_0$ our results are always true.
	\end{remark}

\subsection{Related work}\label{sec:related}
\noindent	For the problem of sampling graphs with a given degree sequence,
	there is an extensive literature, particularly on Markov Chain Monte Carlo (MCMC) methods.
	Jerrum and Sinclair \cite{Jerrum1990} provide an approximate uniform sampler and an approximate counter for  \emph{$P$-stable} degree sequences, for which the  number of graphical realizations of a given degree sequence does not vary too much under small perturbations of the sequence. A first step beyond $P$-stability was recently made by Erd\H{o}s et al. \cite{Erdos2021}.
	Jerrum, Sinclair and Vigoda \cite{Jerrum2004} provide an approximate sampler (and counter) for arbitrary  bipartite degree sequences by reducing the problem to sampling perfect matchings in an appropriate graph representation of the given instance. 
	The work of B\'ezakova, Bhatnagar and Vigoda \cite{Bezakova2007} provides a more direct approach. There are also various non-MCMC methods available in the literature, see, e.g., \cite{Bayati2010,Gao2017,GaoW18,Steger1999,Kim2006,McKay1990}.
	One MCMC approach that has received considerable attention is the \emph{switch Markov chain}, based on the switch operation in Figure \ref{fig:add_del}. This is a simpler, more direct approach than reducing the problem to sampling a perfect matching from a large auxiliary graph. The chain was first analyzed by Kannan, Tetali and Vempala \cite{Kannan1999}, and has been extensively studied,  see, e.g., \cite{Cooper2007,Erdos2013,AK2019,Erdos2022stable}. The state of the art on its mixing time is the work of Erd\H{o}s et al.~\cite{Erdos2022stable}, who show that the chain is rapidly mixing for all $P$-stable degree sequences. 
	
	Rechner et al.~\cite{Rechner2018} introduce the degree interval Markov chain for the \emph{bipartite} version of the problem of sampling graphs with given degree intervals and show its irreducibility  for arbitrary degree intervals.
	Very recently, Erd{\H{o}}s, Mezei and Mikl{\'{o}}s \cite{ErdosMM22} generalized our Theorem \ref{thm:main_chain} to intervals of length $1$ centered around $P$-stable degree sequences. We consider the fact that their meticulous direct approach does not go beyond length $1$ as another indication of the difficulty of directly arguing about the degree interval Markov chain.
	
	The decomposition theorem of Martin and Randall \cite{Martin2006} we use (Theorem \ref{thm:decomposition}), based on the decomposition method of Madras and Randall \cite{Madras2002}, also appeared in an unpublished manuscript by Caracciolo, Pelissetto and Sokal \cite{Caracciolo}.
	Erd\H{o}s et al.~\cite{Erdos2015decomposition} use a related decomposition approach for sampling \emph{balanced joint degree matrix} realizations.

	The result of Liebenau and Wormald \cite{Liebenau2017} builds on a long line of work on asymptotic expressions for the number of graphs with given degrees. Indicatively, Bender and Canfield \cite{BenderC1978} gave a formula for bounded degree sequences and Bollob\'as \cite{Bollobas1980} for $r$-regular sequences with $r = O(\sqrt{\log(n)})$. McKay and Wormald gave expressions both for sparse sequences with maximum degree $o(n^{1/2})$ \cite{McKay1990} and for a certain dense regime \cite{McKay1990conjecture}.
	
	Anari et al. \cite{SLC2}, in a breakthrough recent work,  gave the first polynomial time algorithm for approximate sampling a base of a matroid under a strongly log-concave probability distribution.
	The theory of strongly log-concave (or Lorentzian) polynomials dates back to the work of Gurvits \cite{Gurvits2009}, and was further developed by Anari, Oveis Gharan and Vinzant \cite{AnariGV18} and Br\"and\'en and Huh \cite{BH2019}.
	In another recent work, Kleer \cite{Kleer2021Gibbs} made a connection between asymptotic enumeration formulas and strongly log-concave polynomials for a case of sparse bipartite graphs where only the degrees on one side of the bipartition can vary.
	
\subsection{Outline}
\noindent	In Section \ref{sec:preliminaries} we provide all the necessary preliminaries.  We then continue with a proof overview for Theorems \ref{thm:main_switch_hinge} and \ref{thm:main_chain} in Section \ref{sec:main_result}. For readers with some familiarity regarding Markov chains and some intuition about degree sequence problems, it should be possible to go through (most of) Section \ref{sec:main_result} without delving into (the admittedly long) Section \ref{sec:preliminaries} first. As one of the main building blocks of the proof of Theorem \ref{thm:main_switch_hinge}, we show in Section \ref{sec:slc_lw} that the asymptotic formula of Liebenau and Wormald \cite{Liebenau2017}, when restricted to the degree interval regime of Theorem \ref{thm:main_result}, approximately gives rise to a so-called strongly log-concave polynomial, a result which might be of independent interest.
	Section \ref{sec:decomposition} contains all of the remaining arguments about the Markov chains used in our proofs.

\section{Preliminaries}
\label{sec:preliminaries}
\noindent We need a variety of preliminaries for this work, that are collected in this section (except for the details on the modified log-Sobolev constant, which are deferred to Appendix \ref{sec:mlsc}.)

\subsection{$M$-convexity and strongly log-concave polynomials}
\label{sec:matroid}
\noindent We start with the notion of $M$-convexity for functions \cite{Murota1998,Murota2009}.
Let $\nu : \Z_{ \geq 0}^n \rightarrow \R \cup \{\infty\}$ be a function. The \emph{effective domain} of $\nu$ is given by
$
\text{dom}(\nu) = \{{\bm{\alpha}} \in \N^n : v({\bm{\alpha}}) <  \infty\}.
$
The function $\nu$ is called \emph{$M^\sharp$-convex} if it satisfies the \emph{(symmetric) exchange property}: For any ${\bm{\alpha}}, {\bm{\beta}} \in \text{dom}(\nu)$ and any $i \in [n]$ satisfying $\alpha_i > \beta_i$, there exists a $j \in [n]$ such that $\alpha_j < \beta_j$ and 
\[
\nu({\bm{\alpha}}) + \nu({\bm{\beta}}) \geq \nu({\bm{\alpha}} - \bm{e}_i + \bm{e}_j) + \nu({\bm{\beta}} + \bm{e}_i - \bm{e}_j) \,,
\]
where $\bm{e}_k$ is defined as $\bm{e}_k(\ell) = 1$ if $k = \ell$ and $\bm{e}_k(\ell) = 0$ otherwise.
The function $\nu$ is called \emph{$M$-convex} if it is $M^{\sharp}$-convex \emph{and}  there is an $d \in \N$ such that $\text{dom}(\nu) \subseteq \{{\bm{\alpha}}: \sum_i \alpha_i = d\}$.
A subset $C \subseteq  \Z_{ \geq 0}^n$ is called $M$-convex if the indicator function $\nu_C :  \Z_{ \geq 0}^n \rightarrow \R \cup \{\infty\}$, given by $\nu_C({\bm{\alpha}}) = 1$ if ${\bm{\alpha}} \in C$ and $\nu_C({\bm{\alpha}}) = 0$ otherwise, is $M$-convex.  

We write $\R[x_1,\dots,x_n]$ to denote the set of all polynomials in $x_1,\dots,x_n$ with real coefficients.
We consider polynomials $p \in \R[x_1,\dots,x_n]$ with non-negative coefficients. For a vector ${\bm{\beta}} = (\beta_1,\dots,\beta_n) \in  \Z_{ \geq 0}^n$, we write
$
\partial^{\bm{\beta}} = \prod_{i=1}^n \partial_{x_i}^{\beta_i}
$
to denote the partial differential operator that differentiates a function $\beta_i$ times with respect to $x_i$ for $i = 1,\dots, n$. 
For ${\bm{\alpha}} \in  \Z_{ \geq 0}^n$, we write $x^{\bm{\alpha}}$ to denote $\prod_{i =1}^n x_i^{\alpha_i}$. Furthermore, we write ${\bm{\alpha}}! = \prod_i \alpha_i!$, and for ${\bm{\alpha}}, {\bm{\kappa}} \in  \Z_{ \geq 0}^n$ with $\alpha_i \leq \kappa_i$ for all $i$, we write
$
\binom{{\bm{\kappa}}}{{\bm{\alpha}}} = \prod_{i = 1}^n \binom{\kappa_i}{\alpha_i}.
$
For a constant $c \in \N$ with $c \geq \max_i \alpha_i$, we write $\binom{c}{{\bm{\alpha}}} = \prod_{i = 1}^n \binom{c}{\alpha_i}$.
Let ${\bm{\kappa}} \in  \Z_{ \geq 0}^n$ and the Cartesian product  $K = \times_i \{0,\dots,\kappa_i\}$. Let $w : K \rightarrow \R_{\geq 0}$ be a weight function. The \emph{generating polynomial} of $w$ is
$g_{{\bm{\kappa}}}(x) = \sum_{{\bm{\alpha}} \in K} w({\bm{\alpha}})x^{{\bm{\alpha}}}$.
The \emph{support} of $g_{{\bm{\kappa}}}$ is the set $\text{supp}(g_{{\bm{\kappa}}}) = \allowbreak \{{\bm{\alpha}} \in K : w({\bm{\alpha}}) > 0\}$. 
The  polynomial $g_{{\bm{\kappa}}}$ is called $d$-homogeneous if $|{\bm{\alpha}}| = \sum_i \alpha_i = d$ for all 
$\bm{\alpha} \in \text{supp}(g_{{\bm{\kappa}}}).$

\begin{definition}[Strong log-concavity \cite{Gurvits2009}]
	A polynomial $p \in \R[x_1,\dots,x_n]$  with non-negative coefficients is called \emph{log-concave} on a subset $S \subseteq \R^n_{\geq 0}$ if its  Hessian $\nabla^2 \log(p)$  is negative semidefinite on $S$. A polynomial $p$ is called \emph{strongly log-concave (SLC)} on $S$ if for any ${\bm{\beta}} \in \N^n$, we have that $\partial^{\bm{\beta}} p$ is log-concave.
\end{definition}
\noindent For convenience, the zero polynomial is defined to be SLC always. 
Finally, if the generating polynomial $g_{{\bm{\kappa}}}$ is SLC, then the probability distribution $\pi({\bm{\alpha}}) \propto w({\bm{\alpha}})$ is called SLC as well.
We next state some properties of SLC polynomials that will be used in this work. 

\begin{proposition}[\!\cite{BH2019}]\label{prop:scalar}
	If $p \in \R[x_1,\dots,x_n]$ is SLC and $\gamma \in \R_{\geq 0}$, then $\gamma  p$ is SLC.
\end{proposition}

\begin{proposition}[Following from \cite{BH2019}]\!\!\!\footnote{\ Our $g_{{\bm{\kappa}}}$ is a slight variant the corresponding function of Theorem 3.14 of \cite{BH2019} with $q = 1/e$. The statement for this  $g_{{\bm{\kappa}}}$ follows 
		from a simple transformation of $f_{{\bm{\kappa}}}$ that preserves strong log-concavity, namely the operator that maps $x^{\bm{\alpha}}$ to ${\bm{\alpha}}! \binom{\gamma}{{\bm{\alpha}}}x^{\bm{\alpha}}$ \cite{Branden_communication}.}
	\label{prop:sufficient}
	Let $\nu : \Z_{\geq 0}^n \rightarrow \R \cup \{\infty\}$ with $\text{dom}(\nu) \subseteq \{0,1,\dots,n-1\}^n$ and let
	\begin{equation}\label{eq:standard_lorentzian}
		f_{\bm{\kappa}}(x) = \sum_{{\bm{\alpha}} \in \text{dom}(\nu)} \frac{1}{{\bm{\alpha}}!} e^{-\nu({\bm{\alpha}})}x^{{\bm{\alpha}}} \text{\ \ \ and \ \ \ }  g_{\bm{\kappa}}(x) = \sum_{{\bm{\alpha}} \in \text{dom}(\nu)} \binom{\gamma}{{\bm{\alpha}}} e^{-\nu({\bm{\alpha}})}x^{{\bm{\alpha}}}\,,
	\end{equation}
	be $2m$-homogeneous polynomials, where $\gamma = (n-1,\dots,n-1)$. If $\nu$ is $M$-convex, then $f_{\bm{\kappa}}$ and $g_{\bm{\kappa}}$ are SLC.
\end{proposition}

\subsection{Markov chains and mixing times}
\label{sec:prelim_markov}
\noindent Let $\mathcal{M} = (\Omega,P)$ be an ergodic, time-reversible Markov chain with state space $\Omega$, transition matrix $P$, and stationary distribution $\pi$. We write $P^t(x,\cdot)$ for the distribution over $\Omega$ at time step $t$ with initial state $x \in \Omega$.
The \emph{total variation distance} of this distribution from stationarity at time $t$ with initial state $x$ is
\[\Delta_x(t) = \frac{1}{2}\sum_{y \in \Omega} \big| P^t(x,y) - \pi(y)\big|\,,\]
and the \emph{mixing time} of $\mathcal{M}$ is \[\tau(\epsilon) = \max_{x \in \Omega} \tau_x(\epsilon), \text{\ where\ } \tau_x(\epsilon) = \min\{ t : \Delta_x(t') \leq \epsilon \text{ for all } t' \geq t\} \text{ for } \epsilon>0\,.\]
The chain $\mathcal{M}$  is said to be \emph{rapidly mixing} if its mixing time can be upper bounded by a polynomial in $\ln(|\Omega|/\epsilon)$.

It is well-known that the matrix $P$ only has real eigenvalues $1 = \lambda_0 > \lambda_1 \geq  \dots \geq \lambda_{|\Omega|-1} > -1$.  We may replace $P$ by $(P+I)/2$ to make the chain \emph{lazy}, and hence  guarantee that all its eigenvalues are non-negative. In that case,  by $\gap(P) = 1 - \lambda_1$ we denote the spectral gap of $P$. In this work all Markov chains involved are lazy.
It is well known that one can use the spectral gap to give an upper bound on the mixing time of Markov chain. That is, it holds that
\[\tau_x(\epsilon) \leq \frac{1}{2(1-\lambda_1(P))}\Big( \log \pi(x)^{-1} + 2\log\Big(\frac{1}{2\epsilon} \Big)\Big)\,,\]
as it follows directly from Proposition 1 in \cite{Sinclair1992}.
Furthermore, if one has two Markov chains $\mathcal{M} = (\Omega,P)$ and $\mathcal{M}' = (\Omega,P')$ both with stationary distribution $\pi$ and there are constants $c_1,c_2$ such that $c_1P(x,y) \leq P'(x,y) \leq c_2P(x,y)$ for all $x,y \in \Omega$ with $x \neq y$. Then (see, e.g., \cite{Martin2006}) it follows that
$
c_1\gap(P) \leq \gap(P') \leq c_2\gap(P).
$

The \emph{state space graph} of the chain $\mathcal{M}$ is the directed graph $\mathbb{G}=\mathbb{G}(\mathcal{M})$ with node set $\Omega$
that contains the edges $(x,y) \in \Omega \times \Omega$ for which $P(x,y) > 0$ and $x \neq y$.
that contains an edge $(x,y) \in \Omega \times \Omega$ if and only if $P(x,y) > 0$ and $x \neq y$ (denoted by $x \sim y$). 
Let $\mathcal{P} = \bigcup_{x \neq y} \mathcal{P}_{xy}$, where $\mathcal{P}_{xy}$ is the set of simple paths between $x$ and $y$ in the state space graph $\mathbb{G}$.
A \emph{flow} $f$ in $\Omega$ is a function $\mathcal{P} \rightarrow [0,\infty)$ satisfying
$\sum_{p \in \mathcal{P}_{xy}} f(p) = \pi(x)\pi(y)$ for all $x,y \in \Omega$, $x \neq y$.
The flow $f$ can be extended to the oriented edges $e = (z,z')$ of $\mathbb{G}$ by setting
$f(e) =  \allowbreak \sum_{p \in \mathcal{P} : e \in p } f(p)$,
so that $f(e)$ is the total flow routed through $e\in E(\mathbb{G})$. Let $\textrm{length}(f) = \max_{p \in \mathcal{P} : f(p) > 0} |p|$ be the length of a longest flow-carrying path, and let
$
\textrm{load}(e) \allowbreak = f(e)/Q(e)
$
be the \emph{load} of the edge $e$, where $Q(e) = \pi(x)P(x,y)$ for $e = (x,y)$.
If $\textrm{load}(f) = \max_{e\in E(\mathbb{G})} \psi(e)$ is the maximum load of the flow, it holds that $\gap(P)^{-1} \leq \textrm{load}(f)\,\textrm{length}(f)$  (see, e.g, \cite{Sinclair1992}).

We will sometimes also work (implicitly) with the so-called \emph{modified log-Sobolev constant} $\rho = \rho(P)$. This constant can also be used to upper bound the mixing time of a Markov chain. In particular, it holds that
\[
\tau_x(\epsilon) \leq \frac{1}{\rho(P)}\left( \log \log \pi(x)^{-1} + \log\left(\frac{1}{2\epsilon^2} \right)\right)\,,
\]
see, e.g., \cite{Bobkov2006}. Details on $\rho(P)$ are given in Appendix \ref{sec:mlsc}.

\subsubsection{Markov chain decomposition}\label{sec:prelim_decomp}

\noindent We describe a Markov chain decomposition of Martin and Randall \cite{Martin2006} that follows the decomposition framework of Madras and Randall \cite{Madras2002}. Let $\mathcal{M} = (\Omega,P)$ be a Markov chain and
$\bigcup_{i=1}^q \Omega_i$ be a  partition of $\Omega$ for some $q \in \N$. We define the restriction Markov chains $\mathcal{M}_i = (\Omega_i,P_{\Omega_i})$ as follows. For $x \in \Omega_i$ we let $P_{\Omega_i}(x,y) = P(x,y)$ if $x, y \in \Omega_i$ with $x \neq y$, and $P_{\Omega_i}(x,x) = \allowbreak 1 - \allowbreak \sum_{y \in \Omega_i,y \neq x} P_{\Omega_i}(x,y)$.
Furthermore, let
$
\partial_i(\Omega_j) = \left\lbrace y \in \Omega_j : \exists x \in \Omega_i \text{ with } P(x,y) > 0 \right\rbrace
$
be the set of elements in $\Omega_j$ that can be reached with positive probability in one transition of the chain $\mathcal{M}$ from some element in $\Omega_i$.

Let $\mathcal{M}_{{\textrm{MH}}} = ([q],P_{{\textrm{MH}}})$ be (the Metropolis-Hastings variation of) the projection Markov chain on $[q] = \allowbreak \{1,\dots,q\}$. That is,
$P_{{\textrm{MH}}}(i,j) > 0$ if and only if $\partial_i(\Omega_j) \neq \emptyset$ and, in that case, for $i \neq j$,
\begin{equation}\label{eq:mh_def}
	P_{{\textrm{MH}}}(i,j) = \frac{1}{2\Delta} \min\left\{1,\frac{\pi(\Omega_j)}{\pi(\Omega_i)}\right\} ,
\end{equation}
where $\Delta$ is the maximum out-degree in the state space graph of $\mathcal{M}_{\textrm{MH}}$, while 
\[
P_{\textrm{MH}}(i,i) = 1 -  \sum_{j \in [q] \mysetminus \{i\}} P_{\textrm{MH}}(i,j) \,.
\]
Note that $\mathcal{M}_{\textrm{MH}}$ has stationary distribution $\pi_{{\textrm{MH}}}(i) = \pi(\Omega_i)$ for $i \in \{1,\dots,q\}$ and a holding probability of at least $1/2$.
We will use the following decomposition theorem from \cite{Martin2006}.

\begin{theorem}[\!\cite{Martin2006}, Corollary 3.3]\label{thm:decomposition}
	Suppose there exist $\beta > 0$ and $\gamma >0$ such that $P(x,y) \geq \beta$ for all $x, y$ that are adjacent in $\mathbb{G}(\mathcal{M})$,
	and $\pi(\partial_i(\Omega_j)) \geq \gamma \pi(\Omega_j)$ for all $i, j$ that are adjacent  in  $\mathbb{G}(\mathcal{M}_{{\normalfont{\textrm{MH}}}})$.
	Then
	$\gap(P) \geq \beta\gamma\cdot\gap(P_{{\normalfont{\textrm{MH}}}}) \cdot\min_{i = 1,\dots,q} \gap(P_{\Omega_i})$.
\end{theorem}

\subsubsection{Load-exchange Markov chain}

\noindent In this work, we will need a weighted version of the \emph{base-exchange Markov chain} studied by Anari et al. \cite{SLC2}.
Let $\pi$ be a strongly log-concave probability distribution   with $\pi({\bm{\alpha}}) \propto w({\bm{\alpha}})$ whose support forms an $M$-convex set $C$.
We define the \emph{(unit) load-exchange Markov chain} on $C \subseteq Z_{\geq 0}^n$:

\begin{mdframed}[style=exampledefault]
	Assuming ${\bm{\alpha}} \in C$ is the current state of the \emph{(unit) load-exchange Markov chain}:
	\begin{itemize}[topsep=6pt,itemsep=6pt,parsep=0pt,partopsep=20pt,leftmargin=22pt]
		\item Select an element $i \in [n]$ uniformly at random.
		\item Pick an $\bm{\alpha}' \in C$ with ${\bm{\alpha}}' \geq  {\bm{\alpha}}- {\bm{e_i}}$ with probability $\propto w({\bm{\alpha}}')$ among all such ${\bm{\alpha}}'$.
	\end{itemize}
\end{mdframed}\smallskip

\noindent Similarly to the base-exchange Markov chain \cite{SLC2}, the above procedure defines an ergodic, time-reversible Markov chain with stationary distribution $\pi$ over $C$ given by $\pi({\bm{\alpha}}) \propto w({\bm{\alpha}})$.
Using the notion of \emph{polarization} for SLC polynomials \cite{BH2019}, in combination with a simple Markov chain comparison argument (as in Appendix \ref{app:markov_comparison}), Corollary \ref{cor:polymatroid} can be shown. The proof (which is implicitly given in \cite{Kleer2021Gibbs}), roughly speaking, uses a reduction to the case of matroids, after which a result of Cryan et al.\cite{Cryan2019} gives the desired result.

\begin{corollary}\label{cor:polymatroid}
	Let ${\bm{\kappa}} = (n,\dots,n)$ and suppose that the $d$-homogeneous  polynomial 
	$
	g_{{\bm{\kappa}}}(x) = \sum_{{\bm{\alpha}} \in K} w({\bm{\alpha}})x^{{\bm{\alpha}}} \in \R[x_1,\dots,x_n]
	$
	is SLC.
	Then the transition matrix $P$ of the load-exchange Markov chain on $\mathrm{supp}(g_{{\bm{\kappa}}})$ satisfies  $\rho(P) \geq 1/(n^2d)$, where $\rho(P)$ is the modified log-Sobolev constant of $P$.
\end{corollary}

\subsubsection{Degree intervals and the switch-hinge flip Markov chain.}
\label{sec:chain}
		
\noindent A sequence of non-negative integers $\myd = (d_1,\dots,d_n)$ is called a \emph{graphical degree sequence} if there exists a simple, undirected, labeled graph $G = (V,E)$ on nodes $V = [n]$, where node $i$ has degree $d_i$, for $i \in V$. Such a graph is called a \emph{(graphical) realization of $\myd$}.  By $\mathcal{G}(\myd)$ we denote the set of all graphical realizations of $\myd$, while by $\myd(G)$ we denote the degree sequence of a graph $G$. 
	For given vectors $\bm{\ell} = (\ell_1,\dots,\ell_n)$ and $\bm{u} = (u_1,\dots,u_n)$ of non-negative integers, we define
	$
	\textstyle \mathcal{G}(\bm{\ell},\bm{u}) = \bigcup_{\bm{\ell} \leq \myd \leq \bm{u}} \mathcal{G}(\myd)
	$
	as the set of all graphical realizations $G$ satisfying $\bm{\ell} \leq \myd(G) \leq \bm{u}$, meaning $\ell_i \leq d_i(G) \leq u_i$ for all $i \in V$. For $m \in \N$, we define $\mathcal{G}_m(\bm{\ell},\bm{u})$ as the set of all graphical realizations $G \in \mathcal{G}(\bm{\ell},\bm{u})$ with precisely $m$ edges, i.e., with $\sum_i d_i(G) = 2m$. Finally, we define the set of all degree sequences satisfying the degree interval constraints, and whose total sum of the degrees equals $2m$, as
	$
	\Dm = \left\{\bm{d} : \bm{\ell} \leq \bm{d} \leq u \text{ and } \sum_i d_i = 2m\right\}.
	$
	
	A \emph{fully polynomial almost uniform sampler (FPAUS)} for sampling graphs with given degree intervals $[\bm{\ell},\bm{u}]$ is an algorithm that, for any $\epsilon >0$, outputs a graph $G \in \mathcal{G}(\bm{\ell},\bm{u})$ according to a distribution $\tilde{\pi}$ such that $d_{\mathrm{TV}}(\pi,\tilde{\pi})  \leq \epsilon$,
	where $\pi$ is the uniform distribution over $\mathcal{G}(\bm{\ell},\bm{u})$, and runs in time polynomial in $n, \log(1/\delta)$ and $\log(1/\epsilon)$.
	A \emph{fully polynomial randomized approximation scheme (FPRAS)} for the problem is an algorithm that, for every $\epsilon,\ \delta > 0$, outputs $|\mathcal{G}(\bm{\ell},\bm{u})|$ up to a multiplicative factor $(1 + \epsilon)$ with probability at least $1 - \delta$, in time polynomial in $n$, $1/\epsilon$ and $\log(1/\delta)$.
	Analogous definitions hold for the set
	$\Gmul$ for a given $m$.
	
	Below we define the \emph{switch-hinge flip Markov chain} to uniformly sample elements from $\Gmul$ based on two of the local operations of
	Figure \ref{fig:add_del}. The \emph{degree interval Markov chain} of Theorem \ref{thm:main_chain}, that can also perform addition/deletion operations, is given in Appendix \ref{app:interval_chain} .
	
	\begin{mdframed}[style=exampledefault]
		Assuming $G \in \Gmul$ is the current state of the \emph{switch-hinge flip Markov chain}:
		\begin{itemize}[topsep=6pt,itemsep=6pt,parsep=0pt,partopsep=20pt,leftmargin=22pt]
			\item With probability $2/3$, do nothing.
			\item With probability $1/6$, try to perform a \emph{switch operation}: Choose an ordered tuple of {distinct} nodes $(v,w,x,y)$ uniformly at random. If $\{w,v\}, \{x,y\} \in E(G)$, and $\{y,v\}, \{x,w\} \notin E(G)$, then delete $\{w,v\},\{x,y\}$ from $E(G)$, and add $\{y,v\},\{x,w\}$ to $E(G)$.
			
			\item With probability $1/6$, try to perform a \emph{hinge flip operation}: Choose an ordered tuple of distinct nodes $(v,w,x)$ uniformly at random. If $\{w,v\} \in E(G)$ and $\{w,x\} \notin E(G)$, then delete $\{w,v\}$ from and add $\{w,x\}$ to $E(G)$ if the resulting graph is in $\Gmul$.
		\end{itemize}
	\end{mdframed}
	\smallskip
	Due to the symmetry of the transition probabilities, it is not hard to see that the chain is time-reversible with respect to the uniform  distribution. Also because of the holding probability of at least $2/3$, the chain is aperiodic. Finally, by a simple counting argument, there exists a polynomial $t(n)$ such that $P_{\Gmul}(G, H) \geq 1/t(n)$ for all $G, H \in \Gmul$ with $P_{\Gmul}(G, H)>0$.
	The irreducibility of the chain (i.e., the fact that its state space is strongly connected) for the intervals of interest  will follow implicitly from our analysis (Lemma \ref{lem:partial2}).

\subsection{Near-regular degree sequences}
\label{sec:near-regular}
\noindent Let $r \geq 1$ be a given integer. A degree sequence $\myd$ is said to be $r$-regular if $d_i = r$ for $i \in [n]$. For a fixed $ 0 \leq \alpha < 1$ we say that a degree sequence $\myd$ is \emph{$(\alpha,r)$-near-regular} if $\max_i |d_i - r| \leq r^\alpha$. When we do not refer to a specific $(\alpha,r)$ pair,  we just write about \emph{near-regular} degree sequences.

We state some properties of near-regular degree sequences that we will use later. The most important result is  Theorem \ref{thm:liebenau} below. We use a slightly different formulation than that of \cite{Liebenau2017}.\footnote{\ Our formulation is in line with the note after Conjecture 1.2 in \cite{Liebenau2017}.}
For any degree sequence $\myd = (d_1,\dots,d_n)$, define $\xi =  \sum_i d_i / n$, $\mu = \xi/(n-1)$ and $\chi = \sum_{i} (d_i - \xi)^2 / (n-1)^{2}$. Roughly speaking, the theorem states that if the distance between the  degree sequence $\myd$ and the $\xi$-regular sequence of the same size is not too large, then the expression in \eqref{eq:asymptotic_formula} is a good approximation for $|\Gd|$. The absolute constant $\alpha$ in Theorem \ref{thm:main_result} is mostly restricted by the $\epsilon$ in Theorem \ref{thm:liebenau}.

\begin{theorem}[Liebenau and Wormald \cite{Liebenau2017}]\label{thm:liebenau}
	There exists an absolute constant $\epsilon > 0$ such that for every sequence of degree sequences $\big( \myd^{(n)}\big) _{n \in \N}$ with $\xi n$ even, $\max_{i\in[n]}|d_i^{(n)} - \xi| = o\left( n^\epsilon\min\{\xi,n-\xi-1\}^{1/2}\right)$, and $n^2\min\{\mu,1-\mu\} \rightarrow \infty$, it holds that
	\begin{equation}\label{eq:asymptotic_formula}
		|\Gd| \sim \bar{w}(\bm{d}) := \sqrt{2}\exp\left(\frac{1}{4} - \frac{\chi^2}{4\mu^2(1-\mu)^2}\right)\left(\mu^{\mu}(1 - \mu)^{(1 - \mu)}\right)^{n(n-1)/2} \prod_i\binom{n-1}{d^{(n)}_i}\,. 
	\end{equation}
	To be precise, there exists a non-negative function $\delta(n)$ with $\delta(n) \rightarrow 0$ as $n \rightarrow \infty$, so that the relative error in ``$\sim$'' is bounded above in absolute value by $\delta(n)$ for every such $\big( \myd^{(n)}\big)_{n \in \N} $. 
\end{theorem}
Furthermore, we also rely on the notion of strong stability introduced in \cite{AK2019} (and implicitly already used in \cite{Jerrum1989graphical}). A combinatorial definition of this notion is given below. It essentially states that any graph with a slightly perturbed degree sequence can easily be transformed into a graph with the desired degree sequence by flipping the edges on a short alternating path.
An \emph{alternating $(u,v)$-path in a graph $G$} is a (possibly non-simple) edge-disjoint $(u,v)$-path (in the corresponding complete graph) alternating between edges and non-edges of $G$, starting with an edge adjacent to $u$, and ending with a non-edge adjacent to $v$; recall that a non-edge is an edge contained in the complement of $E(G)$. If $u = v$ we obtain an \emph{alternating cycle}.
To facilitate the definition of strong stability, let $\mathcal{G}'(\myd) = \bigcup_{\myd'} \mathcal{G}(\myd')$ with $\myd'$ ranging over all sequences $\myd'$ satisfying $\sum_i d'_i = \sum_i d_i$ and $\sum_{i} |d'_i - d_i| = 2$, i.e., there exist  $\kappa,\lambda$ such that
$d'_{\kappa} = d_{\kappa} + 1$, $d'_{\lambda} = d_{\lambda} - 1$, and $d'_i =d_i$ otherwise.

\begin{definition}[Strong stability]\label{def:strong}
	A class $\mathcal{D}$ of degree sequences is strongly stable if there exists a constant $k \in \N$, such that for all $\myd \in \mathcal{D}$ and all $G \in \mathcal{G}'(\myd)$, there is an alternating $(u,v)$-path in $G$ of length at most $k$, where $u$ and $v$ are the unique nodes with $\mathrm{deg}_G(u) = d_u + 1$ and $\mathrm{deg}_G(v) = d_v - 1$. 
\end{definition}

\begin{proposition}\label{prop:strong}
	Let $0<\alpha < 1/2$ be a constant and assume that $2 \leq r(n) \leq (1 - \sigma)n$ for some  constant $0 < \sigma < 1$ and $n \in \N$. Then there exists some $n_1\in \N$ so that the class $\mathcal{F}_{(\alpha,r)}$ of $(\alpha,r)$-near-regular degree sequences of length at least $n_1$ is strongly stable.
\end{proposition}

The following two results hold for the class $\mathcal{F}_{(\alpha,r)}$ in Proposition \ref{prop:strong}. Lemma \ref{cor:short_alt_cycle} essentially states that if an edge is present in some graphical realization, then there exists a short alternating cycle to obtain a graphical realization with the same degree sequence not containing that edge. As a result, the subset of realizations in $\mathcal{G}(\myd)$ containing a given edge and the set of realizations not containing it are polynomially related in size.

\begin{lemma}\label{cor:short_alt_cycle}
	Let $\myd \in \mathcal{F}_{(\alpha,r)}$. Suppose that $G \in \Gd$ and let $\{u,v\} \in E(G)$ (resp. $\{u,v\} \notin E(G)$).
	Then there exists a graph $G' \in \Gd$ with $\{u,v\} \notin E(G')$ (resp.~$\{u,v\} \in E(G)$) and
	{$E(G) \triangle E(G')$ is an alternating cycle of length at most $12$.}
	Similarly, suppose that $\{u,w\}, \{u,v\} \in E(G)$. Then there exists a graph $G' \in \Gd$ with $\{u,w\} \in E(G')$ and $\{u,v\} \notin E(G')$, and
	{$E(G) \triangle E(G')$ is an alternating cycle of length at most $12$.}
\end{lemma}

Furthermore, 
the \emph{switch Markov chain} is rapidly mixing for the class $\mathcal{F}_{(\alpha,r)}$. This follows directly from \cite{AK2019} where it is shown that the switch Markov chain is rapidly mixing for all strongly stable classes of degree sequences. In particular, we will use the following result.

\begin{corollary}[Follows from \cite{AK2019}]\label{cor:switch} 
	Let $q(n) \geq 2$ be a given polynomial  and consider the \emph{lazy switch Markov chain} $\mathcal{M} = (\Gd,P_{\Gd})$ for some $\myd \in \mathcal{F}_{(\alpha,r)}$ that proceeds as follows: For a given $G \in \Gd$
	\begin{itemize}[topsep=6pt,itemsep=6pt,parsep=0pt,partopsep=20pt,leftmargin=22pt]
		\item With probability $1 - 1/q(n)$ do nothing, and
		\item With probability $1/q(n)$, try to perform a switch operation.
	\end{itemize}
	Then there exists a polynomial $p(n)$, such that for any $\myd \in \mathcal{F}_{(\alpha,r)}$ we have $\gap(P_{\Gd}) \geq 1/p(n)$.\footnote{\ Note that $P_{\Gd}$ depends on $r(n)$.}
\end{corollary}

\section{Proof approach overview}
\label{sec:main_result}
\noindent In this section we give a high-level overview of the proofs of Theorems \ref{thm:main_switch_hinge} and \ref{thm:main_chain}. The idea is to decompose the degree interval Markov chain twice, using the graph operations in Figure \ref{fig:add_del}. (The second step suffices to prove Theorem \ref{thm:main_switch_hinge}, and both steps are needed to prove Theorem \ref{thm:main_chain}.) We first decompose $\Gul$ based on the addition/deletion operation. Every part of the decomposition corresponds to a set $\Gmul$ containing all graphs respecting the degree intervals $[\bm{\ell},\bm{u}]$ and having exactly $m$ edges, for some $m$. That is, there is a one-to-one correspondence between the possible values of $m$, and the parts of the decomposition.
The Markov chain decomposition result of Theorem \ref{thm:decomposition} tells us that if the switch hinge-flip chain is rapidly mixing for every $m$, and if it relatively ``easy'' to move between the different parts $\Gmul$ by means of additions/deletions, then the degree interval chain is rapidly mixing on $\Gul$. In the second step we carry out a similar decomposition, but now based on the hinge flip operation. That is, for a given $m$ we decompose $\Gmul$ in the sets $\Gd$ for all sequences $\myd$ which satisfy the interval constraints, and whose degrees sum up to $2m$. If the switch chain is rapidly mixing on every $\Gd$, and we can move ``easily'' between the sets $\Gd$ using hinge flip operations, then the switch hinge-flip Markov chain on $\Gmul$ is also rapidly mixing. We continue with a formalization of these statements.

Let $\mathcal{T} = \{m_1,\dots,m_2\}$, where $m_1$ and $m_2$ are the minimum and maximum number of edges, respectively, that any $G \in \mathcal{G}(\bm{\ell},\bm{u})$ could have; e.g., $m_1 = \frac{1}{2}\sum_{i} \ell_i$ and $m_2 = \frac{1}{2}\sum_{i} u_i$ in case the two summands are even.

First we partition $\mathcal{G}(\bm{\ell},\bm{u})$ into disjoint sets $\mathcal{G}_{m}(\bm{\ell},\bm{u})$ for $m \in \mathcal{T}$. Recall that
$\mathcal{G}_m(\bm{\ell},\bm{u}) = \big\{G \in \mathcal{G}(\bm{\ell},\bm{u}) : \sum_{i} d_i(G) = 2m\big\}$.
The restriction Markov chains $\mathcal{M}_{\mathcal{G}_m(\bm{\ell},\bm{u})}$ are essentially given by restricting the original chain to only perform switch and hinge flip operations that respect the degree intervals on graphs with precisely $m$ edges.
Applying Theorem \ref{thm:decomposition}---with $\beta$ and $\gamma$ to be determined later---we get
\begin{equation}\label{eq:decomp1}
	\gap(P) \geq \beta\gamma\cdot\gap(P_{\mathcal{T}}) \cdot\min_{m \in \mathcal{T}} \gap(P_{\mathcal{G}_m(\bm{\ell},\bm{u})}) \,,
\end{equation}
where $P_{\mathcal{T}}$ is the transition matrix of the Metropolis-Hastings projection chain on $\mathcal{T}$, and $P$ the transition matrix of the degree interval Markov chain. The goal will be to show that $\beta$ and $\gamma$, as well as all the spectral gaps, can be lower bounded by an inverse polynomial function of the form $1/p(n)$ for some polynomial $p(n)$. This means that $\gap(P)$ is lower bounded by an inverse polynomial as well, which is equivalent to showing that the degree interval Markov chain is rapidly mixing (see Section \ref{sec:prelim_markov}).

Next we  partition the sets $\mathcal{G}_m(\bm{\ell},\bm{u})$ further into sets $\mathcal{G}(\myd)$ for sequences $\myd$ in $\mathcal{D}_m(\bm{\ell},\bm{u}) = \big\{\myd : \sum_{i} d_i = 2m \allowbreak \text{ and } \bm{\ell} \leq \myd \leq \bm{u} \big\}$.
For this part of the decomposition we get a Metropolis-Hastings projection chain on the set $\mathcal{D}_m$.
The restriction chains on $\mathcal{M}_{\mathcal{G}(\myd)}$ are the chains in which we essentially only apply switch operations on all graphs with degree sequence $\myd$. This is precisely the switch Markov chain with some polynomially bounded holding probability (as defined in Corollary \ref{cor:switch}).
Using one more time Theorem \ref{thm:decomposition} for each $m$---again, with $\beta_m$ and $\gamma_m$ to be determined later---we have
\begin{equation}\label{eq:decomp2}
	\gap(P_{\mathcal{G}_m(\bm{\ell},\bm{u})}) \geq \beta_m\gamma_m\cdot\gap(P_{\mathcal{D}_m}) \cdot \min_{\myd \in \mathcal{D}_m} \gap(P_{{\Gd}}) \,,
\end{equation}
where $P_{\mathcal{D}_m}$ is the transition matrix of the Metropolis-Hastings chain on $\mathcal{D}_m$. This time, in order to show that the switch-hinge flip Markov chain is rapidly mixing, we need to bound $\gamma_m,\, \beta_m$,\, and all the spectral gaps by an inverse polynomial function.

Combining (\ref{eq:decomp1}) and (\ref{eq:decomp2}) we now get
\begin{equation}\label{eq:decomp3}
	\gap(P) \geq  \beta \gamma  \cdot\gap(P_{\mathcal{T}})\cdot \min_{m \in \mathcal{T}} \Big\{\beta_m \gamma_m \cdot\gap(P_{\mathcal{D}_m}) \cdot\min_{\myd \in \mathcal{D}_m} \gap(P_{\Gd})\Big\} \,,
\end{equation}
and in order to show that the degree interval Markov chain is rapidly mixing, we need to show that $\beta$, $\gamma$, all $\beta_m$ and $\gamma_m$, and all spectral gaps can be lower bounded by $1/q(n)$ for some polynomial  $q(n)$.

While this is what we are going to do for Theorem \ref{thm:main_chain}, recall that for Theorem \ref{thm:main_switch_hinge} (directly) and Theorem \ref{thm:main_result} (through Theorem \ref{thm:main_switch_hinge} and the reductions in Appendix \ref{sec:counting}), we only show that the switch-hinge flip Markov chain is rapidly mixing. In that case, it suffices to show that $\beta_m$,  $\gamma_m$, and the spectral gaps involved in \eqref{eq:decomp2} are polynomially bounded for any given $m$ (and the polynomial bound is independent of $m$), i.e., we only need to globally consider the second decomposition step. 
A polynomial lower bound on $\beta$ and each one of the $\beta_m$ follows by the very definition of the degree interval Markov chain (see also the discussion after its definition).
In order to bound $\gamma$ and $\gamma_m$, for all $m$, we use Lemma \ref{cor:short_alt_cycle}. Roughly speaking, we need to show that we can move rather easily between realizations of two degree sequences $\myd$ and $\myd'$, with $\sum_i |d_i - d_i'| = 2$. The high-level idea for $\gamma_m$ is to show that it is either directly possible to perform a hinge flip in order to transition from a graph $G$ with degree sequence $\myd$ to some $G'$ with degree sequence $\myd'$, or that $G$ is not too far away from some other graph $H$ with the same degree sequence $\myd$ from which it is  possible to directly  move to some $G'$ with degree sequence $\myd'$ via a hinge flip. We take an analogous approach for bounding $\gamma$ but in terms of addition/deletion operations rather than hinge flips.

The gaps of the chains $\mathcal{M}_{\Gd}$ are globally bounded because of known rapid mixing results for the switch Markov chain \cite{AK2019} (Corollary \ref{cor:switch}).
Therefore, in order to show Theorem \ref{thm:main_switch_hinge}, it remains to bound $\gap(P_{\mathcal{D}_m})$, which we do in Appendix \ref{sec:D_m}; an outline is given in Section \ref{sec:FPRAS} below. For Theorem \ref{thm:main_chain}, we additionally need to bound $\gap(P_{\mathcal{T}})$, which we do in Appendix \ref{sec:T}; a brief outline is given in Section \ref{sec:main_chain}.

\subsection{Proving Theorem \ref{thm:main_switch_hinge}}
\label{sec:FPRAS}
\noindent The main technical challenge of Theorem \ref{thm:main_switch_hinge} lies in proving that the resulting Metropolis-Hastings projection chain on $\Dm$ is rapidly mixing, i.e., that $\gap(P_{\mathcal{D}_m})$ can be polynomially bounded.  We sometimes refer to this chain as the \emph{hinge flip projection chain}. Note that for $\bm{d},\bm{d}' \in \Dm$, with $||\bm{d} - \bm{d}'||_1 = 2$, it follows from \eqref{eq:mh_def} that 
	\[
	P_{{\textrm{MH}}}(\bm{d},\bm{d}') \ge \frac{1}{2n^2} \min\left\{1,\frac{|\mathcal{G}(\bm{d}')|}{|\Gd|}\right\} \,,
	\]
by taking the obvious upper bound $\Delta \le n^2$ in \eqref{eq:mh_def}. So, intuitively, whether or not the hinge flip projection chain is rapidly mixing depends on the quantities $|\Gd|$ for $\myd \in \Dm$. To this end, we first argue, using a comparison argument, that if suffices to show that the load-exchange Markov chain on $\Dm$, i.e., the Markov chain that allows us to move between degree sequences by adjusting the degree of two nodes by $1$ (while keeping the degree sums fixed), is rapidly mixing for the weights $w(d) = |\Gd|$. (It is not hard to see that $\mathcal{D}_m$ is in fact an $M$-convex set.) A Markov chain comparison argument, very informally speaking, proceeds by showing that if one Markov chain is rapidly mixing, and a second chain is very close to it (in terms of similar stationary distribution and transition probabilities), then the second chain is also rapidly mixing. In our setting, the comparison is based on the fact that both chains have the same stationary distribution $\pi$ with $\pi(\bm{d}) \propto w(\bm{d})$, and the fact that their transition probabilities are polynomially related for the degree sequences that we are interested in (using the remark at the end of Appendix \ref{app:markov_comparison}).

In order to show that the load-exchange Markov chain on $\Dm$ is rapidly mixing, we would like to use Corollary \ref{cor:polymatroid}, which states that the load-exchange Markov chain is rapidly mixing if a polynomial identified with its stationary distribution satisfies the property of strong log-concavity (SLC). To be precise, we may apply Corollary \ref{cor:polymatroid} if, for given $\bm{\ell}, \bm{u}$ and $m$, the polynomial
\[
h(x) = \sum_{\bm{d} \in \mathcal{D}_m(\bm{\ell},\bm{u})} w(\bm{d}) \cdot x^{\bm{d}} = \sum_{\bm{d} \in \mathcal{D}_m(\bm{\ell},\bm{u})} |\Gd|\cdot x^{\bm{d}}
\]
is SLC. This seems hard to prove (and might not be true in general). However, it turns out that when replacing the weights $w(\bm{d})$ by  by the approximations $\bar{w}(\bm{d})$ as given in the asymptotic formula \eqref{eq:asymptotic_formula} of Liebenau and Wormald \cite{Liebenau2017}, the resulting polynomial
\begin{align}\label{eq:approx_SLC}
	\bar{h}(x) = \sum_{\bm{d} \in \mathcal{D}_m(\bm{\ell},\bm{u})} \bar{w}(\bm{d})\cdot x^{\bm{d}}
\end{align}
is in fact SLC, when considering the degree interval instances of Theorem \ref{thm:main_result}.\footnote{\ We remark at this point, that, although this is true for the regime considered in Theorem \ref{thm:main_result}, this does not seem to be true for the general range in Theorem \ref{thm:liebenau}.}
We show this fact in Theorem \ref{thm:approx_slc} in Section \ref{sec:slc_lw} by observing that the polynomial in \eqref{eq:approx_SLC} is of the form \eqref{eq:standard_lorentzian} in Proposition \ref{prop:sufficient}, which is a general sufficient condition for a polynomial to be SLC  \cite{BH2019}.

The above implies that if we run the load-exchange Markov chain with the approximations $\bar{w}(\bm{d})$, it is in fact rapidly mixing with stationary distribution $\bar{\pi}$ given by $\bar{\pi}(\bm{d}) \propto \bar{w}(\bm{d})$. Now, the approximations have the property that for some $n_0$ sufficiently large, it holds that for all $n \geq n_0$ and $\bm{d} \in \Dm$,
\[
e^{-(19/\sigma)^2}|\Gd| \leq \bar{w}(\bm{d}) \leq |\Gd|  \,.
\]
This also implies that
\[
e^{-(19/\sigma)^2}\pi(\bm{d}) \leq \bar{\pi}(\bm{d}) \leq e^{(19/\sigma)^2}\pi(\bm{d}) \,.
\]
One can then again use a Markov chain comparison argument to argue that the load-exchange Markov chain based on the original weights $w(\bm{d})$ is also rapidly mixing (based on the final remark in Appendix \ref{app:markov_comparison} with $\delta = 1 - e^{-(19/\sigma)^2}$). This in turn implies that the hinge flip  projection chain is also rapidly mixing, which is what we wanted to show.

\bigskip\noindent\textbf{General framework.}
The approach described above for showing rapid mixing of the switch-hinge flip Markov chain might be applicable to other classes of degree interval instances. Informally speaking, the essential things that are needed are the following two things:
\begin{enumerate}
	\item The degree sequences satisfying the interval constraints are strongly stable (see Def.~\ref{def:strong}).
	\item The weights $|\Gd|$ ``approximately'' give rise to an  SLC polynomial.
\end{enumerate}

The requirement of strong stability in the first point is needed for various reasons. First of all, it is a sufficient condition for the switch Markov chain (i.e., the restrictions chains in our decomposition) to be rapidly mixing \cite{AK2019}. Secondly, we rely on it when bounding the parameter $\gamma$ in the Martin-Randall decomposition theorem (Theorem \ref{thm:decomposition}). Thirdly, strong stability is sufficient to argue that the transition probabilities of the  load-exchange Markov chain, and the Metropolis-Hastings projection chain, are polynomially related (so that we can use a Markov chain comparison argument to compare their mixing times).

For the second point, even if the weights $|\Gd|$ do not give rise to an  SLC polynomial, one may still make things work. It suffices to find values $z(\bm{d})$ and polynomials $q_1$ and $q_2$ such that
\[
\frac{1}{q_1(n)}|\Gd| \leq z(\bm{d}) \leq q_2(n)|\Gd|  \,,
\]
and for which
\[
\bar{h}(x) = \sum_{\bm{d} \in \mathcal{D}_m(\bm{\ell},\bm{u})} z(\bm{d})\cdot x^{\bm{d}}
\]
is SLC.

\subsection{Proving Theorem \ref{thm:main_chain}}
\label{sec:main_chain}
\noindent In order to prove Theorem \ref{thm:main_chain}, we additionally need to show that the projection chain on $\mathcal{T}$ is also rapidly mixing, i.e., that $\gap(\mathcal{T})$ is polynomially bounded. In other words, we consider the Metropolis-Hastings projection Markov chain with state space $\{a,\dots,b\}$, where
$a = \frac{1}{2} \sum_i \ell_i$ and $b = \frac{1}{2}\sum_i u_i$ (where we assume $a$ and $b$ to be integral here for the sake of simplicity),
and $\pi(m) \propto |\Gmul|$ for $m \in \{a,\dots,b\}$. This chain will sometimes be referred to as the \emph{addition/deletion projection chain}. A sufficient condition for this Markov chain to be rapidly mixing is that the sequence $(w_m)_{m = a,\dots,b}$ given by $w_m = |\Gmul|$ is log-concave, meaning that for every $m$
$
w_{m+1}w_{m-1} \leq w_m^2.
$
We show log-concavity of this sequence to be true when the intervals have size at most one (corresponding to the statement of Theorem \ref{thm:main_chain}) by using a variation on an argument of Jerrum and Sinclair \cite{Jerrum1989}.

\begin{remark}
	One might wonder if the theory of strongly log-concave polynomials can also be used to prove rapid mixing for degree intervals beyond size one. For example, one might consider a polynomial
	of the form
	\[
	g(x_1,\dots,x_n,y) = \sum_{m = m_1}^{m_2} \sum_{\bm{d} \in \mathcal{D}_m(\bm{\ell},\bm{u})} \binom{n-1}{2m_2 - 2m}\bar{z}(\bm{d})\cdot y^{2m_2-2m}x^{\bm{d}} \,,
	\]
	where $m_1$ and $m_2$ are the minimum and maximum number of edges that any graph in $\mathcal{D}(\ell,u)$ can have, respectively. This is then a $2m_2$-homogeneous polynomial.
	The problem that now occurs though, is that the domain of this polynomial, indexed by the tuples $(d_1,\dots,d_n,2m_2-2m)$, {can be shown not to be an $M$-convex set} and, thus, $g$ cannot be SLC.
\end{remark}

\section{SLC property in a restricted range of the  Liebenau-Wormald result}
\label{sec:slc_lw}
Throughout this section, we consider $m$ and $n$ as fixed.  Recall that for a given degree sequence $\bm{d} = (d_1,\dots,d_n)$ we defined $\xi = \xi(n,m) =  \sum_i d_i / n = 2m/n$, $\mu = \mu(n,m) = \xi/(n-1) = 2m/(n(n-1))$ and  $\chi(\bm{d}) = \sum_{i} (d_i - \xi)^2 / (n-1)^{2}$. Furthermore, in  \eqref{eq:asymptotic_formula} we defined
\begin{equation}\label{eq:lieb_approx}
	\bar{w}(\bm{d}) = \sqrt{2}\exp\left(\frac{1}{4} - s(\bm{d})^2\right)\left(\mu^{\mu}(1 - \mu)^{(1 - \mu)}\right)^{n(n-1)/2}\prod_i\binom{n-1}{d_i} \,,
\end{equation}
which is approximately the number of graphs with degree sequence $\myd$ in case it is near-regular. Here we have
\[
s(\bm{d}) = \frac{\chi(\bm{d})}{2\mu(1-\mu)} \,.
\]
In order to show that the weights $\bar{w}(\bm{d})$ give rise to an SLC polynomial,  \[\bar{h}(x) = \sum_{\bm{d} \in \mathcal{D}_m(\bm{\ell},\bm{u})} \bar{w}(\bm{d})\cdot x^{\bm{d}}\,,\]
it would be sufficient by Proposition \ref{prop:sufficient} to argue that the function $s(\bm{d})^2$ is $M$-convex.
Unfortunately, it turns out that this is not the case.
Instead, we simply show that in the regime of Theorem \ref{thm:main_result}, it holds that $s(\bm{d}) = O(1)$, and so we can ignore the contribution $\exp(-s(\bm{d})^2)$ in \eqref{eq:lieb_approx} at the expense of a slightly worse bound on the mixing time. The resulting approximation formula is easily seen to be SLC, which intuitively follows from a discrete form of log-concavity of the binomial coefficients; see Theorem \ref{thm:approx_SLC}.

\begin{lemma}
	Under the conditions on $[\bm{\ell},\bm{u}]$ as in Theorem \ref{thm:main_result}, with $0 < \sigma < 1$ and $2 \leq r = r(n) \leq (1- \sigma)n$, if  $n$ is  large enough, then for any $\bm{\ell} \leq \bm{d} \leq\bm{u}$ it holds that
	\begin{equation}\label{eq:s_bound}
		0 \leq s(\bm{d}) \leq \frac{18 n}{\sigma  r^{1- 2\alpha} (n-1)}\, \leq\, \frac{19}{\sigma} \,.
	\end{equation}
	\label{lem:lieb_simple}
\end{lemma}
\begin{proof}
	By the definition of $\chi(\bm{d})$, $s(\bm{d}) \geq 0$ always holds. To see the upper bound, let $n_1= \max\Big\lbrace \Big\lceil( 2/\sigma)^\frac{1}{1-2\alpha}\Big\rceil , \allowbreak \lceil 18/\sigma\rceil \Big\rbrace$ and note that the quantity $s(\bm{d})$ can be rewritten as
	\begin{equation}\label{eq:s_bound2}
		s(\bm{d}) = \frac{\chi{(\bm{d})}}{2\mu(1 - \mu)} = \frac{n^2(n-1)^2}{(n-1)^2} \frac{\sum_i (d_i - \xi)^2}{2\cdot 2m(n(n-1)-2m)}\,,
	\end{equation}
	where $2m = \sum_i d_i$.
	Note that $\sum_i (d_i - \xi)^2 \leq n (2r^{\alpha})^2 = 4nr^{2\alpha}$. Moreover, we can bound $m$ using the simple facts that $r-r^{\alpha}\ge r/4$ and $n^{\alpha} \le \sigma n/2 $ for $n\ge n_1$. The latter implies that
	\[r+r^{\alpha}\leq (1-\sigma)n+(1-\sigma)n^{\alpha}\le (1-\sigma)n+ \sigma n/2 = (1-\sigma/2)n \,.\]
	So, we have
	\[n \frac{r}{4}  \le 2m \le n \left(1-\frac{\sigma}{2}\right) n\,,\]
	and therefore,
	\begin{equation}\label{eq:s_denominator}
		2m(n(n-1)-2m)   \ge n \frac{r}{4} \left(1 - \left(1-\frac{\sigma}{2}\right)\frac{n}{n-1}\right)n(n-1)
		\geq \frac{r \sigma n^2 (n-1)}{9} \,,
	\end{equation}
	where the last inequality holds because $\frac{n\sigma/2 -1}{n-1}\ge \frac{4\sigma}{9}$ for $n\ge n_1$.
	By combining \eqref{eq:s_bound2} and \eqref{eq:s_denominator}, we then get
	\[
	s(\bm{d}) \leq \frac{n^2 \cdot 4nr^{2\alpha} \cdot 9}{2 r \sigma n^2 (n-1)} = \frac{18n}{\sigma  r^{1- 2\alpha} (n-1)} \,,
	\]
	which completes the second inequality. The final inequality holds because $r \geq 2$ and $n/(n-1) \leq \frac{19}{18}$ for $n\ge n_1$. 
\end{proof}

We next summarize the main result of this section, and give the remaining small technical steps of its proof. In a nutshell, it states that a simplified version of the Liebenau-Wormald formula which is within a constant factor from the original in \eqref{eq:asymptotic_formula} is approximately SLC in the regime of Theorem \ref{thm:main_result}.

\begin{theorem}\label{thm:approx_SLC}
	For given $n ,m \in \N$, $\bm{\ell},u \in \N^n$ with $\bm{\ell} \leq \bm{u}$, and degree sequence $\bm{d}$ with $\sum_i d_i = 2m$ and $\bm{\ell} \leq \bm{d} \leq \bm{u}$, let
	\begin{equation}\label{eq:lieb_approx2}
		\bar{z}(\bm{d}) = \sqrt{2}e^{\frac{1}{4}}\left(\mu^{\mu}(1 - \mu)^{(1 - \mu)}\right)^{n(n-1)/2}\prod_i\binom{n-1}{d_i} \,.
	\end{equation}
	The resulting $2m$-homogeneous polynomial
	\[
	\bar{f}(x) = \sum_{\bm{d} \in \mathcal{D}_m(\bm{\ell},\bm{u})} \bar{z}({\bm{d}})\cdot x^{\bm{d}
	}\]
	is SLC.
	
	Furthermore, there exists an $n_0 \in \N$ such that for all $n \geq n_0$ and $m \geq n$, if the degree interval $[\bm{\ell},\bm{u}]$ satisfies the conditions of Theorem \ref{thm:main_result}, then
	\begin{equation}\label{eq:lieb_approx3}
			e^{-(19/\sigma)^2}|\Gd| \leq \bar{z}(\bm{d}) \leq |\Gd|\,,
	\end{equation}
	for every $\bm{\ell} \leq \bm{d} \leq \bm{u}$.
	\label{thm:approx_slc}
\end{theorem}
\begin{proof}
	We first note that the factors $\sqrt{2}$ and $\left(\mu^{\mu}(1 - \mu)^{(1 - \mu)}\right)^{n(n-1)/2}$ can all be seen as non-negative scalars as $n$ and $m$ are given. This means, by Proposition \ref{prop:scalar}, that it suffices to show that  the polynomial with coefficients
	\[
	e^{\frac{1}{4}}\prod_i\binom{n-1}{d_i}
	\]
	is SLC.
	
	Comparing this to the second polynomial in Proposition \ref{prop:sufficient}, it follows that we can simply choose $\nu$ to be the constant function $\nu(\bm{d}) = -\frac{1}{4}$ on its effective domain $\Dm$. (As mentioned earlier, intuitively it is SLC because the binomial coefficients satisfy a discrete form of log-concavity.) 
	The statement in \eqref{eq:lieb_approx3} follows directly from Theorem \ref{thm:liebenau} and Lemma \ref{lem:lieb_simple}.
\end{proof}

\section{Decomposition of the degree interval Markov chain}
\label{sec:decomposition}
\noindent In this section we give the missing details regarding the decomposition steps as outlined in Section \ref{sec:main_result}.

\subsection{Bounding $\beta_m, \gamma_m$ and $\gap(P_{\mathcal{D}_m})$}
\label{sec:D_m}
\noindent Throughout this section we assume that some $m \in \{m_1,\dots,m_2\}$ is fixed. Moreover, recall that we consider degree intervals of the form $[d_i,d_{i}+1]$, or $[d_i,d_i]$, for $i \in[n]$. It is not hard to see that $\beta_m\ge (6n^4)^{-1}$. This rough polynomial bound follows directly from the transition probabilities of the degree interval Markov chain.

We first lower bound the $\gamma_m$ in Lemma \ref{lem:partial2} below. By the definition of the hinge flip operation 
we have that for any $\myd,\myd' \in \mathcal{D}_m$, there is a strictly positive transition probability between $\myd$ and $\myd'$ if and only if $\sum_{i} |d_i - d'_i| = 2$.

The proof of Lemma \ref{lem:partial2} follows from Lemma \ref{cor:short_alt_cycle}, where it is shown that for a graph with a given degree sequence, we can always find a graph with a slightly perturbed degree sequence that is close to the former in terms of symmetric difference (when the original sequences satisfies strong stability).

\begin{lemma}\label{lem:partial2}
	There exists a polynomial $q_1(n)$ such that, for any feasible $m$ and for all $\myd,\myd' \in \mathcal{D}_m$ with $\sum_{i} |d_i - d'_i| = 2$, we have
	$
	\pi_{\mathcal{D}_m}\left( \partial_{\myd}\left( \mathcal{G}\left( \myd'\right) \right) \right)  \geq \frac{1}{q_1(n)}\pi_{\mathcal{D}_m}\left( \mathcal{G}\left( \myd'\right) \right) .
	$
\end{lemma}

\begin{proof}
	Again assume that $n \ge n_1= \Big\lceil\left( 2/\sigma\right)^\frac{1}{1-2\alpha}\Big\rceil$.
	Let $a$ and $b$ be the unique nodes such that $d_a' = d_a + 1$ and $d_b' = d_b - 1$; note that the uniqueness of $a, b$ follows from the condition $\sum_{i} |d_i - d_i'| = 2$. Let
	\[
	\mathcal{H} = \{G \in \mathcal{G}(\myd) : \exists c \in [n] \text{ such that } \{b,c\} \in E(G), \{a,c\} \notin E(G)\} \,,
	\]
	and note that it has the property
	\begin{equation}
		\label{eq:dd1}
		|\partial_{\myd}(\mathcal{G}(\myd'))| \geq \frac{1}{n}|\mathcal{H}|\,.
	\end{equation}
	To see this, note that for a given $G \in \mathcal{H}$, we can perform the hinge flip that removes the edge $\{b,c\}$ and adds the edge $\{a,c\}$ to obtain an element in $\mathcal{G}(\myd')$. Moreover, there can be at most $n$ graphs $G \in \mathcal{H}$ that map onto a given $G' \in \partial_{\myd}(\mathcal{G}(\myd'))$, as there are at most $n$ choices for $c$.
	
	Moreover, using the second part of Lemma \ref{cor:short_alt_cycle}, we show that
	\begin{equation}
		\label{eq:dd2}
		|\mathcal{H}| \geq \frac{1}{n^{12}}|\mathcal{G}(\myd)|\,.
	\end{equation}
	To see this, note that for any $G \in \mathcal{G}(\myd)$, we have $d_b = d_b' + 1 \geq 0$ which implies that $b$ has at least one neighbor $c$ in $G$. Now, if $\{a,c\} \notin E(G)$ we obtain an element in $\mathcal{H}$; otherwise, we can find a graph $G'$ close to $G$ (using Lemma \ref{cor:short_alt_cycle}) for which $\{a,c\} \notin E(G)$ while still $\{b,c\} \in E(G)$. As there are at most $n^{12}$ graphs $G \in \mathcal{G}(\myd)$ that map to the same $G' \in \mathcal{H}$, the inequality \eqref{eq:dd2} follows. Moreover, we also have $n^{10} |\mathcal{G}(\myd)| \geq |\mathcal{G}(\myd')|$ which follows directly from Definition \ref{def:strong} and (the proof of) Proposition \ref{prop:strong} (see Appendix \ref{app:near-regular}).
	
	Combining the last observation with \eqref{eq:dd1} and \eqref{eq:dd2} then yields
	\[
	|\partial_{\myd}(\mathcal{G}(\myd'))| \geq \frac{1}{q_1(n)} |\mathcal{G}(\myd')|\,,
	\]
	for $q_1(n)=n^{23}$. Dividing both sides by $\sum_{\myd \in \mathcal{D}_m} |\mathcal{G}(\myd)|$, then gives the desired result.
\end{proof}

It remains to bound $\gap(P_{\mathcal{D}_m})$. As explained in Section \ref{sec:FPRAS}, the first step is to carry out a comparison argument with the load-exchange Markov chain with weights $w(\bm{d}) = |\Gd|$ (so that it will be sufficient to study the mixing time of the latter). Remember that both {the hinge flip projection chain}, 
as well as the load-exchange chain  have stationary distribution $\pi(\bm{d}) \propto w(\bm{d})$. 

In what follows we write $\mathcal{M}_{\mathcal{D}_m} = (\mathcal{D}_m,P)$ for the (Metropolis-Hastings) hinge flip projection chain, and $\mathcal{M}'_{\mathcal{D}_m} = (\mathcal{D}_m,P')$ for the load-exchange chain on $\mathcal{D}_m$.

\begin{lemma}
	\label{lem:compare_load_mh}
	There exists a polynomial $p(n)$ such that
	\[
	p(n) \gap(P_{\mathcal{D}_m}) \geq \gap(P'_{\mathcal{D}_m}) \,.
	\]
	for any $m \in \{m_1,\dots,m_2\}$.
\end{lemma}
\begin{proof}[Proof (sketch)]
	It suffices to show that there exists polynomials $p_1$ and $p_2$ such that, whenever $\bm{d},\bm{f} \in \mathcal{D}_m$ satisfy $||\bm{d} - \bm{f}||_1 = 2$, then
	\begin{equation}
		\frac{1}{p_1(n)} \leq P(\bm{d},\bm{f}), P'(\bm{d},\bm{f}) \leq \frac{1}{p_2(n)} \,.
		\label{eq:poly_relation}
	\end{equation}
	This then implies that the transition probabilities $P(\bm{d},\bm{f})$ and $P'(\bm{d},\bm{f})$ are themselves polynomially related. In turn, 
	this implies the existence of the desired polynomial $p(n)$ as both chains have the same stationary distribution and therefore their spectral gaps are polynomially related (see Appendix \ref{app:markov_comparison}).
	
	The existence of the polynomials in \eqref{eq:poly_relation} follows from the fact that all near-regular degree sequences are strongly stable. First of all, it holds that the term $|\mathcal{G}(\bm{d}')|/|\Gd|$ in the Metropolis-Hastings hinge flip projection chain can always be upper and lower bounded by a polynomial because of strong stability. Furthermore, in the load-exchange Markov chain we pick (in the second step) a new degree sequence $\bm{d}'$ proportional to $w(\bm{d}')$ over all possible choices of $\bm{d}'$ with $||\bm{d} - \bm{d}'||_1 = 2$ that respect the degree interval bounds. For a given $\bm{d}$, let $N(\bm{d})$ be the set of all such sequences $\bm{d}'$. Then the probability of transitioning to $\bm{d}'$ is (up to an additional polynomial factor because of the first step of the load-exchange Markov chain) equal to
	\[
	\frac{|\mathcal{G}(\bm{d}')|}{\displaystyle\sum_{\bm{f} \in N(\bm{d})} |\mathcal{G}(\bm{f})|}\,,
	\]
	which can again be upper and lower bounded by a polynomial because of strong stability, and because $|N(\bm{d})| \leq n^2$.
\end{proof}

Lemma \ref{lem:compare_load_mh} implies that we may focus on bounding $\gap(P'_{\mathcal{D}_m})$. Now, by the arguments given in Section \ref{sec:FPRAS} in combination with another simple comparison argument using \eqref{eq:mlsc_comparison} and Theorem \ref{thm:approx_slc}, it suffices to bound $\gap(P''_{\mathcal{D}_m})$ where $P''$ is the transition matrix of the hinge flip Markov chain in which we replace the weights $w(\bm{d})$ by the approximations $\bar{z}(\bm{d})$ as in \eqref{eq:lieb_approx2}.

In Section \ref{sec:slc_lw}, we showed that the polynomial in \eqref{eq:lieb_approx3} is in fact SLC, so then Corollary \ref{cor:polymatroid} implies that the modified log-Sobolev constant of this chain can be lower bounded by a polynomial, which implies the same for the spectral gap by \eqref{eq:spectral_sobolev}. This completes this section, and shows in particular that the switch-hinge flip Markov chain is rapidly mixing, which in turn completes the proof of Theorem \ref{thm:main_switch_hinge}.

\subsection{Bounding $\beta, \gamma$ and $\gap(P_{\mathcal{T}})$}
\label{sec:T}
\noindent Recall that $\mathcal{M}_{\mathcal{T}}$ is the Metropolis-Hastings chain on the index set $\mathcal{T} = \{m_1,\ldots,\allowbreak m_2\}$. For simplicity, we use
$
w_m = |\mathcal{G}_m(\bm{\ell},\bm{u})|
$
to denote the number of feasible graphical realizations with $m$ edges. Note that for any $m \in \mathcal{T}$ we have
$\pi_{\mathcal{T}}(m) = w_m/\sum_{i \in \mathcal{T} } w_i$,
and that $P_{\mathcal{T}}(m,m') > 0$ if and only if $|m - m'| \leq 1$. From the definition of the degree interval Markov chain, it immediately follows that $\beta \geq 1/q(n)$ for some polynomial $q(n)$. We  lower bound $\gamma$ in the following lemma following the same approach as for Lemma \ref{lem:partial2}.

\begin{lemma}\label{lem:partial1}
	There exists a polynomial $q_2(n)$ such that, for all $m,m' \in \mathcal{T}$ with $|m - m'| = 1$, we have
	$
	\pi_{\mathcal{T}}\left( \partial_m\left( \mathcal{G}_{m'}\right) \right)  \geq \frac{1}{q_2(n)}\pi_{\mathcal{T}}\left( \mathcal{G}_{m'}\right)
	$.
\end{lemma}

\begin{proof}
	Assume that $m' = m + 1$ (the case $m' = m-1$ is analogous). Let $G \in \mathcal{G}_{\myd}$ for some $\myd \in \mathcal{D}_m$. Note that $m' \geq m_1 + 1 > m_1$, which implies that there are nodes $i$ and $j$ whose degrees in $G$ are not equal to the upper bound of their degree interval.
	Note that the set
	\[
	\mathcal{H} = \{G \in \mathcal{G}_{m} : \exists i,j \in [n] \text{ with } d_i(G) <  u_i, d_j(G) < u_j \text{ and } \{i,j\} \notin E(G) \}
	\]
	has the property that
	\begin{equation}\label{eq:T_1}
		|\partial_m(\mathcal{G}_{m'})| \geq \frac{1}{m+1} |\mathcal{H}|\,.
	\end{equation}
	In order to see this, note that for any $G \in \mathcal{H}$ we can add the edge $\{i,j\}$ (recall that these nodes depend on the choice of $G$) to obtain an element in $\mathcal{G}_{m'}$. On the other hand,there can be at most $m+1$ graphs $G$ that map onto a given graph $H \in \mathcal{G}_{m'}$ using this procedure. This gives the inequality \eqref{eq:T_1}.
	
	Moreover, using the first part of Lemma \ref{cor:short_alt_cycle} and following the same argument as in the proof of Lemma \ref{lem:partial2}  it can be shown that
	\begin{equation}\label{eq:T_2}
		|\mathcal{H}| \geq \frac{1}{n^{12}}|\mathcal{G}_{m}|\,.
	\end{equation}
	To see this, note that for any graph $G \in \mathcal{G}_{m}$, nodes  $i$ and $j$ with $d_i(G) < u_i$ and $d_j(G) < u_j$ always exist, as $m < m_2$. Moreover, if $\{i,j\} \in E(G)$ we know from Lemma \ref{cor:short_alt_cycle} that there is a graph $G'$ with the same degree sequence not containing edge $\{i,j\}$ close to $G$.

	We next show that $|\mathcal{G}_{m}| \geq |\mathcal{G}_{m'}|/p(n)$ for some polynomial $p(n)$. To see this, note that for any $G' \in |\mathcal{G}_{m'}|$ there exist nodes $x$ and $y$ such that $d_x(G') > \ell_x$ and $d_y(G') > \ell_y$ as $m' = m + 1 > m_1$. If $\{x,y\} \in E(G')$ we can remove it to obtain an element of $|\mathcal{G}_{m}|$. Otherwise, again using Lemma \ref{cor:short_alt_cycle} we can first find an element $G'' \in \mathcal{G}_{m'}$ close to $G'$ that contains $\{x,y\}$ and the remove it. Combining this with \eqref{eq:T_2} yields the existence of a polynomial $q_2(n)$ such that
	\[
	|\mathcal{H}| \geq \frac{1}{q_2(n)}|\mathcal{G}_{m'}|\,.
	\]
	Finally, combining the latter inequality with \eqref{eq:T_1} and dividing both sides by $\sum_{m \in \mathcal{T}} w_m$, then gives the desired result.
\end{proof}

In order to show that $\mathcal{M}_{\mathcal{T}}$ is rapidly mixing or, in particular, that the gap $\gap(\mathcal{T})$ can be polynomially bounded, it is sufficient to show that the sequence $(w_m)_{m \in \mathcal{T}}$ is \emph{log-concave}. Log-concavity means that for any $m \in \mathcal{T} \mysetminus \{m_1,m_2\}$, $w_{m-1}w_{m+1} \leq w_m^2$.

\begin{theorem}\label{thm:log-concavity}
	The sequence $(w_m)_{m \in \mathcal{T}}$ is log-concave for all interval sequences $[\bm{\ell},\bm{u}]$ for which $u_i \in \{\ell_i,\ell_i+1\}$ for all $i \in[n]$.
\end{theorem}

\begin{proof}
	We follow the notation, terminology and general outline of the proof of Theorem 5.1 in \cite{Jerrum1989}.
	Define $A = \mathcal{G}_{m+1} \times \mathcal{G}_{m-1}$ and $B = \mathcal{G}_{m} \times \mathcal{G}_{m}$. We will show that $|A| \leq |B|$, from which the claim follows.
	
	Note that the symmetric difference of any two subgraphs of $K_n$ can be decomposed into a collection of alternating cycles and simple paths. We will do this in a canonical way.\footnote{\ This decomposition is the main extra step we need compared to the proof of Theorem 5.1 in \cite{Jerrum1989}. The symmetric difference of two matchings is by construction already a disjoint union of cycles and paths. This is also where the analysis breaks down in case the degree intervals have length at least two.} Fix some total order $\preceq_e$ on the edges of $K_n$. For two subgraphs $G$ and $G'$ we will call edges in $E(G) \mysetminus E(G')$ blue, and edges in $E(G')\mysetminus E(G)$ red. Around every node, we will pair up blue edges with red edges as much as possible. We do this by repeatedly selecting a node and pairing up the lowest ordered red and blue edge that have not yet been paired up.  This yields a decomposition of the symmetric difference into i) alternating red-blue cycles, ii) alternating simple red-blue paths of even length (with same number of red and blue edges),  iii) simple paths ending and starting with a red edge, iv) simple paths ending and starting with a blue edge. We call this the \emph{canonical symmetric difference decomposition of $E(G)\triangle E(G')$ with respect to $\preceq_e$}, or simply the canonical decomposition of $E(G) \triangle E(G')$. We call a simple path a $G$-path if it contains one more edge of $G$ than of $G'$ (i.e., red edges are one more than blue edges),
	and a $G'$-path if it contains one more edge of $G'$.  We emphasize that any path of odd length in the symmetric difference is of one of these two types.
	
	Now, for every pair $(G,G') \in A$ it holds that the number of $G$-paths exceeds the number of $G'$-paths by precisely two (as $G$ has two edges more than $G'$). For this reason, we partition $A$ into disjoint classes $\{A_r : r = 1,\dots,m\}$ where
	\[
	\begin{array}{ll}
		A_r = \{(G,G') \in A : &\text{ the canonical decomposition of } E(G)\triangle E(G')  \\
		&  \text{ contains } r+1\ G\text{-paths and } r-1\ G'\text{-paths}\}\,.
	\end{array}
	\]
	In order to prove $|A| \leq |B|$ it suffices to show $|A_r| \leq |B_r|$ for all $r$. We call a pair $(L,L') \in B$ \emph{reachable} from $(G,G')\in A$ if and only if $E(G) \triangle E(G) = E(L) \triangle E(L')$ and $L$ is obtained by from $G$ by taking some $G$-path in the canonical decomposition and flipping the parity
	of the edges with respect to $G$ and $G'$. It is important to see that the canonical symmetric difference decomposition of the pairs $(G,G')$ and $(L,L')$ is the same because all degree intervals have length one.  Note that the number of pairs in $B_r$ reachable from a given $(G,G') \in A_r$ is precisely the number of $G$-paths in the canonical decomposition of $G$ and $G'$, which is $r+1$. Conversely, any given $(L,L') \in B_r$ is reachable from precisely $r$ pairs in $A_r$. Therefore, if $|A_r| > 0$, we have
	\[
	\frac{|B_r|}{|A_r|} = \frac{r+1}{r} > 1\,.
	\]
	This proves the claim.
\end{proof}

We are ready to bound  the spectral  gap of $P_{\mathcal{T}}$. Note that $\Omega$ in the statement of Theorem \ref{thm:mixing_concave} is actually $\mathcal{T}$. Recall that $|\mathcal{T}| =m_2-m_1+1\le n/2 +1 \le n$. Moreover, the ratios $w_i/w_j$  are also polynomially bounded for any $i,j \in \mathcal{T}$ with $|i - j| = 1$.
This can be shown exactly as in the proofs of the Lemmata \ref{lem:partial2} and \ref{lem:partial1}; see also Appendix \ref{sec:counting}.
As a result, it is sufficient to prove the statement in Theorem \ref{thm:mixing_concave} below in order to bound the gap of $P_{\mathcal{T}}$.

\begin{theorem}\label{thm:mixing_concave}
	Let $(w_m)_{m \in \Omega}$ be a log-concave sequence of non-negative numbers and let $\mathcal{M} = (\Omega,P)$ be a Markov chain 
	with transition probabilities
	\[
	P(i,j) = \left\{ \begin{array}{ll}
		\frac{1}{4} \min\left\{1, w_j/w_i\right\} & \text{ if } |i - j| = 1\,, \\
		0 & \text{ if } |i - j| > 1\,, \\
		1 - P(i,i-1) - P(i,i+1) & \text{ if } i = j\,.
	\end{array}\right.
	\]
	Then $1/\gap(P) \leq 4 \,|\Omega|^3 \max_{i,j : |i-j| = 1}  w_i/w_j$.\footnote{\ We suspect a similar result is true without the dependence on the $w_i$ but this is not needed for our purpose.}	
\end{theorem}

\begin{proof}
	First note that the stationary distribution $\pi$ of $\mathcal{M}$  is proportional to the weights $(w_i)_{i \in \Omega}$, i.e., $\pi(i) = w_i/\sum_{p \in \Omega} w_p$, as desired.
	We bound the congestion of the straightforward multi-commodity flow $f$ in which we route $\pi(i)\,\pi(j)$ units of flow over the path $i\to (i+1)\to \dots \to j$ if $i < j$, or $i \to (i-1) \to \dots \to j$ if $i > j$.
	
	We consider a fixed transition $e = (z,z+1)$. Note that the proof for transitions of the form $(z,z-1)$ is symmetric, since a sequence $(w_i)_{i \in \Omega}$ is log-concave if and only if the sequence $(w_{|\Omega|-i+1})_{i \in \Omega}$ is log-concave.
	We have
	\begin{align}
		\textrm{load}(e) &= \sum_{1 \leq i \leq z} \sum_{z < j \leq |\Omega|} \frac{\pi(i)\,\pi(j)}{\pi(z)P(z,z+1)} \leq 4 \max_{i,j : |i-j| = 1} \frac{w_i}{w_j} \sum_{1 \leq i \leq z} \sum_{z < j \leq |\Omega|} \frac{\pi(i)\,\pi(j)}{\pi(z)} \nonumber\\
		&= 4 \max_{i,j : |i-j| = 1} \frac{w_i}{w_j} \,\bigg(\sum_{p \in \Omega} w_p \bigg)^{-1} \!\!\!\sum_{1 \leq i \leq z} \sum_{z < j \leq |\Omega|} \frac{w_iw_j}{w_z} \,. \label{eq:double_sum}
	\end{align}
	Log-concavity of the sequence $(w_q)_{q \in \Omega}$ implies that for any fixed $i < j$, and any $a \in \N$ such that $i + a \leq j - a$, we have
	\begin{equation}\label{eq:strong_log}
		w_iw_j \leq w_{i+a}w_{j-a} \,.
	\end{equation}
	This follows from repeatedly applying the log-concavity condition. Indeed, log-concavity gives us
	$\frac{w_i}{w_{i+1}}\le \frac{w_{i+1}}{w_{i+2}}\le \ldots \le \frac{w_{j-2}}{w_{j-1}}\le \frac{w_{j-1}}{w_{j}}$
	and thus $w_iw_j \leq w_{i+1}w_{j-1}$. By repeating this with $i+1$ and $j-1$ (i.e., by removing the outer terms) we get
	$\frac{w_{i+1}}{w_{i+2}}\le \ldots \le \frac{w_{j-2}}{w_{j-1}}$
	and thus $w_{i+1}w_{j-1}\leq w_{i+2}w_{j-2}$. After $a$ steps we get \eqref{eq:strong_log}.

	Now, for a fixed $i$ and $j$ in the double summation  in \eqref{eq:double_sum}, let $a_{ij}$ be such that $w_{i+a_{ij}}$ or $w_{j-a_{ij}}$ (or both) equals $w_z$.  Then \eqref{eq:strong_log} gives us that
	$w_iw_j \leq w_z\,w_p$
	for some $p \in \Omega$. Note that for any choice of $z$, the double summation in \eqref{eq:double_sum} has at most $|\Omega|^2$ terms (as there are at most $|\Omega|$ choices for $i$ and $j$). This implies that
	\[
	\sum_{1 \leq i \leq z} \sum_{z < j \leq |\Omega|} w_i w_j / w_z \leq |\Omega|^2 \sum_{p \in \Omega}  w_z w_p/w_z = |\Omega|^2 \sum_{p \in \Omega} w_p \,.
	\]
	Combining this inequality with  \eqref{eq:double_sum}, we obtain
	\[
	\textrm{load}(e) \leq 4|\Omega|^2\max_{i,j : |i-j| = 1} w_i/w_j  \,.
	\]
	The proof is completed once we notice that, by the definition of the flow $f$, $\textrm{length}(f) \leq |\Omega|$.
\end{proof}

This then completes the proof of Theorem \ref{thm:main_chain}.

\section{Discussion and future directions}
\label{sec:disc}
\noindent We did not attempt to optimize the upper bounds on the mixing times of the Markov chains involved.
Already for the \emph{switch Markov chain} no low-degree polynomial upper bounds are known on its mixing time. For instance, the best known upper bound for $r$-regular graphs is $r^{23}n^8(rn\log(rn) + \log(1/\epsilon))$  \cite{Cooper2007,Cooper2012corrigendum}. This is a central  issue for many MCMC approaches for sampling graphs with given degrees (or degree intervals in our case). Various non-MCMC approaches to the problem, see, e.g., \cite{Bayati2010,Gao2017,GaoW18,Steger1999,Kim2006,McKay1990}, often have better running times,  but only work for smaller classes of degree sequences or have weaker guarantees on the uniformity of the output than we require in our setting.

An interesting first direction for future work is determining whether the degree interval chain is rapidly mixing for more general instances. 
The most intriguing question from our point of view, however, is whether there is a black-box reduction implying that if the switch Markov chain is rapidly mixing for all degree sequences $\myd$ satisfying $\bm{\ell} \leq \myd \leq \bm{u}$, then the degree interval Markov chain is also rapidly mixing. Even more generally, can the problem of sampling graphs with given degree intervals always be reduced to the problem of sampling graphs with given degrees?

Further, one could explore other, non-MCMC, approaches for approximate sampling, especially when the degree ranges are relatively large.
Can one come up with an algorithm in which resampling certain ``bad events'' (e.g., resampling edges adjacent to a node not satisfying its degree interval constraints) yields an exactly uniform sample, following the ``partial rejection sampling'' framework of Guo, Jerrum and Liu \cite{Guo2019}? While this seems unlikely when sampling graphs with given degrees, we suspect it is possible for the problem of sampling graphs with (sufficiently large) given degree intervals.

\bibliographystyle{plainurl}
\bibliography{references_STACS}

\newpage
\appendix

\section{Missing details from Section \ref{sec:preliminaries}}
\label{app:missing_from_prelim}
\noindent We give the full definition of the degree interval Markov chain and the two missing proofs from Section \ref{sec:preliminaries}.

\subsection{Degree interval Markov chain}
\label{app:interval_chain}

\noindent The \emph{degree interval Markov chain} to uniformly sample elements from $\mathcal{G}(\bm{\ell},\bm{u})$, based on the local operations of
Figure \ref{fig:add_del}, is defined as follows.

\begin{mdframed}[style=exampledefault]
	Assuming $G \in \mathcal{G}(\bm{\ell},\bm{u})$ is the current state of the \emph{degree interval Markov chain}:
	\begin{itemize}[topsep=6pt,itemsep=6pt,parsep=0pt,partopsep=20pt,leftmargin=22pt]
		\item With probability $1/2$, do nothing.
		\item Otherwise:
		\begin{itemize}[topsep=3pt,itemsep=6pt,parsep=0pt,partopsep=20pt,leftmargin=18pt]
			\item With probability $1/6$, try to perform a \emph{switch operation}: choose an ordered tuple of distinct nodes $(v,w,x,y)$ uniformly at random. If $\{w,v\}, \{x,y\} \in E(G)$, and $\{y,v\}, \{x,w\} \notin E(G)$, then delete $\{w,v\},\{x,y\}$ from $E(G)$, and add $\{y,v\},\{x,w\}$ to $E(G)$.
			
			\item With probability $1/6$, try to perform a \emph{hinge flip operation}: choose an ordered tuple of distinct nodes $(v,w,x)$ uniformly at random. If $\{w,v\} \in E(G)$ and $\{w,x\} \notin E(G)$, then delete $\{w,v\}$ from and add $\{w,x\}$ to $E(G)$ if the resulting graph is in $\mathcal{G}(\bm{\ell},\bm{u})$.
			
			\item With probability $1/6$, try to perform an \emph{addition/deletion operation}: select an ordered tuple of distinct nodes $(v,w)$ uniformly at random. If $\{v,w\} \in E(G)$, then delete it from $E(G)$ if the resulting graph is in $\mathcal{G}(\bm{\ell},\bm{u})$. Similarly, if $\{v,w\} \notin E(G)$, then add it to $E(G)$ if the resulting graph is in $\mathcal{G}(\bm{\ell},\bm{u})$.
		\end{itemize}
	\end{itemize}
\end{mdframed}
\smallskip

\noindent Due to the symmetry of the transition probabilities, it is not hard to see that the chain is time-reversible with respect to the uniform  distribution. Also because of the holding probability of at least $1/2$, the chain is aperiodic. Finally, by a simple counting argument, there exists a polynomial $t(n)$ such that $P_{\mathcal{G}(\bm{\ell},\bm{u})}(G, H) \geq 1/t(n)$ for all $G, H \in \mathcal{G}(\bm{\ell},\bm{u})$ with $P_{\mathcal{G}(\bm{\ell},\bm{u})}(G, H)>0$.
The irreducibility of the chain (i.e., the fact that its state space is strongly connected) for the intervals of interest  will follow implicitly from our analysis, in particular Lemmata \ref{lem:partial2} and \ref{lem:partial1}.

In order to obtain the  \emph{switch-hinge flip Markov chain} for the uniform sampling from $\Gmul$, one may simply ignore addition/deletion operations, resulting in the description in Section \ref{sec:preliminaries}, which is given here as well as a quick reference. This explains the holding probability of $2/3$.

\begin{mdframed}[style=exampledefault]
	Assuming $G \in \Gmul$ is the current state of \emph{switch-hinge flip Markov chain on $\Gmul$}:
	\begin{itemize}[topsep=6pt,itemsep=6pt,parsep=0pt,partopsep=20pt,leftmargin=22pt]
		\item With probability $2/3$, do nothing.
		\item With probability $1/6$, perform a \emph{switch operation}.
		\item With probability $1/6$, perform a \emph{hinge flip operation}.
	\end{itemize}
\end{mdframed}

\smallskip

\subsection{Missing proofs from Subsection \ref{sec:near-regular}}
\label{app:near-regular}
\noindent For convenience, before giving the proofs here, we first repeat the corresponding statements. \medskip

\begin{rtheorem}{Proposition}{\ref{prop:strong}}
	Let $0<\alpha < 1/2$ be a constant and assume that $2 \leq r(n) \leq (1 - \sigma)n$ for some  constant $0 < \sigma < 1$ and $n \in \N$. Then there exists some $n_1\in \N$ so that the class $\mathcal{F}_{(\alpha,r)}$ of $(\alpha,r)$-near-regular degree sequences of length at least $n_1$ is strongly stable.
\end{rtheorem}

\begin{proof}
	Let $n \ge n_1= \Big\lceil\left( 2/\sigma\right)^\frac{1}{1-2\alpha}\Big\rceil$. It is then a matter of simple calculations to verify that the condition $(d_{\max} - d_{\min} + 1)^2 \leq 4d_{\min}(n - d_{\max} - 1)$
	is satisfied for all $\myd \in \mathcal{F}_{(\alpha,r)}$, where $d_{\min}$ and $d_{\max}$ are the minimum and maximum value of $\myd$, respectively. Sequences satisfying this condition are strongly stable  for $k = 10$ \cite{Jerrum1989graphical,AK2019}.
\end{proof}

\begin{rtheorem}{Lemma}{\ref{cor:short_alt_cycle}}
	Let $\myd \in \mathcal{F}_{(\alpha,r)}$. Suppose that $G \in \Gd$ and let $\{u,v\} \in E(G)$ (resp. $\{u,v\} \notin E(G)$).
	Then there exists a graph $G' \in \Gd$ with $\{u,v\} \notin E(G')$ (resp.~$\{u,v\} \in E(G)$) and
	$E(G) \triangle E(G')$ is an alternating cycle of length at most $12$.
	Similarly, suppose that $\{u,w\}, \{u,v\} \in E(G)$. Then there exists a graph $G' \in \Gd$ with $\{u,w\} \in E(G')$ and $\{u,v\} \notin E(G')$, and
	$E(G) \triangle E(G')$ is an alternating cycle of length at most $12$.
\end{rtheorem}
\begin{proof}
	Assume $n \ge n_1= \Big\lceil\left( 2/\sigma\right)^\frac{1}{1-2\alpha}\Big\rceil$ as in the proof of Proposition \ref{prop:strong}.
	Note that, in all cases below, the degree sequence $\myd$ itself plays the role of being a perturbed degree sequence in the argument. By inspecting the proof of Proposition \ref{prop:strong}, this is allowed since $n$ here is sufficiently large.

	For the first case of the first part (i.e., when $\{u,v\}\in E(G)$), let $y$ be such that $\{y,u\} \notin E(G)$. Such a non-edge is guaranteed to exist, as $n\ge n_1 > 2/\sigma$ and the maximum degree of any node will then be bounded away from $n-2$. Also note that $y$ has degree at least $2$.
	By Proposition \ref{prop:strong}, we know that there exists some alternating $(y,v)$-path of length at most $10$. Combining this path with the non-edge $\{y,u\}$ and the edge $\{u,v\}$, results in an alternating cycle of length at most $12$. Hence, if we flip all the edges on this alternating cycle, we obtain a $G' \in \mathcal{G}(\myd)$ with the desired property.
	
	For the second case of the first part (i.e., when $\{u,v\}\notin E(G)$), we pick a $y$  such that $\{y,u\} \in E(G)$. By Proposition \ref{prop:strong}, we consider some alternating $(v,y)$-path (of length at most $10$). Combining this path with the edge $\{y,u\}$ and the non-edge $\{u,v\}$, we again obtain an alternating cycle of length at most $12$. By flipping  this cycle, we get a $G' \in \mathcal{G}(\myd)$ with the desired property.
	
	\medskip

	For the second part of the lemma, we make a similar, albeit more complicated, argument. We distinguish two cases and then consider subcases depending on the relative position of the edge $\{u,w\}$ with respect to some alternating path.
	\smallskip
	
	\noindent \underline{Case 1:}  \emph{there exists $y$ such that $\{y,u\}, \{y,v\} \notin E(G)$.} Consider such a node $y$. By Proposition \ref{prop:strong}, there exists an alternating $(y,v)$-path of length at most $10$.
	A key observation here is that this alternating $(y,v)$-path might contain the edge $\{u,w\}$. Of course, if this is not true, we proceed like in the first case of the first part above. So assume that the alternating $(y,v)$-path does contain $\{u,w\}$.
	If $\{u,w\}$ goes from $w$ to $u$ as we traverse the path from $y$ to $v$, then by taking the remaining $(u,v)$-subpath of the alternating path together with the edge $\{u,v\}$ we get an alternating cycle (of length at most $8$), that contains the edge $\{u,v\}$ but not the edge $\{u,w\}$.
	If $\{u,w\}$ goes from $u$ to $w$ as we traverse the path from $y$ to $v$, then by taking the $(y,u)$-subpath preceding $\{u,w\}$ on the alternating path together with the edge $\{u,v\}$ and the non-edge $\{y,v\}$ we again get an alternating cycle (of length at most $10$), that contains  $\{u,v\}$ but not  $\{u,w\}$.
	
	In any case, by flipping the edges on the corresponding alternating cycle, we obtain a $G' \in \mathcal{G}(\myd)$ with the desired property.
	\smallskip
	
	\noindent \underline{Case 2:} \emph{for every $y$ such that $\{y,u\} \notin E(G)$ we have $\{y,v\} \in E(G)$.} Consider a node $x$ such that $\{x,v\} \notin E(G)$. Given the assumption of the current case, it must be that $\{x,u\} \in E(G)$.
	By Proposition \ref{prop:strong}, there exists an alternating $(x,u)$-path of length at most $10$. If the edge $\{u,w\}$ is not contained in this alternating path, then the whole path together with the non-edge $\{x,v\}$ and the edge $\{u,v\}$ is an alternating cycle (of length at most $12$) that contains  $\{u,v\}$ but not  $\{u,w\}$. Now, assume that the alternating $(x,u)$-path contains $\{u,w\}$.
	If $\{u,w\}$ goes from $u$ to $w$ as we traverse the path from $x$ to $u$, then by taking the $(x,u)$-subpath preceding $\{u,w\}$ on the alternating path together with the edge $\{u,v\}$ and the non-edge $\{x,v\}$ we  get an alternating cycle (of length at most $8$), that contains $\{u,v\}$ but not  $\{u,w\}$.
	Finally, if $\{u,w\}$ goes from $w$ to $u$ as we traverse the path from $x$ to $u$, then by taking the remaining $(u,u)$-subpath of the alternating path together with the edges $\{u,v\}, \{x,u\}$ and the non-edge $\{x,v\}$, we again get an alternating cycle (of length at most $10$), that contains the edge $\{u,v\}$ but not the edge $\{u,w\}$.
	
	In all subcases, by flipping the edges on the corresponding alternating cycle, we obtain a $G' \in \mathcal{G}(\myd)$ with the desired property.
\end{proof}

\newpage

\section{Modified log-Sobolev constant}
\label{sec:mlsc}
\noindent Let $\mathcal{M} = (\Omega,P)$ be a time-reversible Markov chain with stationary distribution $\pi$, and $f,g : \Omega \rightarrow \R_{\geq 0}$. Let
$\mathbb{E}_{\pi}(f) = \sum_{x \in \Omega} \pi(x)f(x)$.
Furthermore, define the entropy-like quantity
\[
\text{Ent}_{\pi}(f) = \mathbb{E}_{\pi}\left[f\log(f) - f\log(\mathbb{E}_{\pi}(f))\right]\,,
\]
and the \emph{Dirichlet form}
\begin{align*}
	\mathcal{E}_P(f,g) &= \frac{1}{2}\sum_{x \in \Omega} \sum_{y \in \Omega} \pi(x)P(x,y)[f(x) - f(y)][g(x) - g(y)]\,.
\end{align*}
The \emph{modified log-Sobolev constant} of the Markov chain $\mathcal{M}$ is defined by
\[
\rho(P) = \inf\left\{ \frac{\mathcal{E}_P(f,\log(f))}{\text{Ent}_{\pi}(f)} \ \Big| \ f: \Omega \rightarrow \R_{\geq 0}, \ \text{Ent}_{\pi}(f) \neq 0 \right\}.
\]
As stated in Section \ref{sec:preliminaries}, it holds that (see, e.g., \cite{Bobkov2006})
\begin{align*}
	\tau_x(\epsilon) \leq \frac{1}{\rho(P)}\left( \log \log \pi(x)^{-1} + \log\left(\frac{1}{2\epsilon^2} \right)\right).
\end{align*}
Furthermore, for any Markov chain it holds that
\begin{equation}
	\label{eq:spectral_sobolev}
	2(1- \lambda_1(P)) \geq \rho(P) \,,
\end{equation}
where $\lambda_1(P)$ is the second-largest eigenvalue of $P$ (assuming the Markov chain is lazy).

\subsection{Markov chain comparison}\label{app:markov_comparison}
\noindent A useful property of proving mixing time bounds through the modified log-Sobolev constant, is that it is easy to see that small perturbations in the transition probabilities and the stationary distribution only result in mild variations in the modified log-Sovolev constant of the resulting Markov chain (by means of a Markov chain comparison argument). Goel \cite{Goel2004} states the following for the modified log-Sobolev constant, based on similar results for the other constants by Diaconis and Saloff-Coste \cite{Diaconis1996}. The notation $W(\Omega,\pi)$ is used to denote the set of all (test) functions $f : \Omega \rightarrow \R_{\geq 0}$.

\begin{theorem}[Lemma 4.1 \cite{Goel2004}]\label{thm:markov_comparison} Let $\mathcal{M} = (\Omega,P)$ and $\mathcal{M}' = (\Omega',P')$ be two finite, reversible Markov chains with stationary distributions $\pi$ and $\pi'$, respectively, and modified log-Sobolev constant $\rho$ and $\rho'$, respectively.
	Assume there is a mapping $\phi: W(\Omega,\pi) \rightarrow W'(\Omega',\pi')$ mapping $f \rightarrow f'$ for $f : \Omega \rightarrow \R_{\geq 0}$, and constants $C, c > 0$ and $B \geq 0$ such that for all $f \in W(\Omega,\pi)$, we have
	\[
	\mathcal{E}_{P'}(f',\log f') \leq C \cdot \mathcal{E}_{P}(f, \log f) \ \ \text{ and } \ \ c \cdot \text{Ent}_{\pi}(f) \leq \text{Ent}_{\pi'}(f') + B\cdot \mathcal{E}_{P}(f, \log f)\,.
	\]
	Then
	\[
	\frac{c \rho'}{C + B\rho'} \leq \rho \,.
	\]

\end{theorem}
\noindent In particular, if $\Omega = \Omega'$ and there exists a $0 < \delta < 1$ such that $(1-\delta)P(x,y) \leq P'(x,y) \leq (1+\delta)P(x,y)$ for all $x,y \in \Omega$, and $(1-\delta)\pi(x) \leq \pi'(x) \leq (1 + \delta)\pi(x)$ for $x \in \Omega$, it directly follows that
\begin{equation}
	\label{eq:mlsc_comparison}
	\frac{1}{\rho} \leq \frac{1+\delta}{1-\delta}\cdot \frac{1}{\rho'} \,.
\end{equation}

\newpage

\section{Reductions for approximate sampling and counting}
\label{sec:counting}
We first explain how Theorem \ref{thm:main_result} follows from Theorem \ref{thm:main_switch_hinge}.
The induced approximate sampler from Theorem \ref{thm:main_switch_hinge} can be turned into an approximate counter for $|\Gmul|$ by standard techniques (Section \ref{app:sm_to_cm}). Furthermore, this approximate counter can then be turned into an approximate counter for $|\Gul|$ by means of a simple reduction (Section \ref{app:cm_to_c}). In turn, the approximate counter for $|\Gul|$ can be transformed into an approximate sampler from $\Gul$, again by a standard technique (Section \ref{app:c_to_s}).

A subtle point is that it is not known whether the problem of sampling and counting from $\Gul$, or $\Gmul$, is \emph{self-reducible} \cite{JerrumVV86}. This follows roughly from the fact that it is not known whether the problem of sampling/counting from $\Gd$ is self-reducible in general. However, if one restricts to degree intervals $[\bm{\ell},\bm{u}]$ for which both an FPRAS and FPAUS for $\Gd$ is known for every $\bm{\ell} \leq \bm{d} \leq \bm{u}$, like the set of P-stable degree sequences \cite{Jerrum1990},  then standard reduction techniques for self-reducible problems can still be applied.

To give a more concrete intuition, in the reduction of Section \ref{app:sm_to_cm} the problem is that one needs to be able to compute the final factor in the telescoping product \eqref{eq:telescope} efficiently, which we do not know how to do for an  arbitrary degree sequence $u$ (although we do know it for $P$-stable degree sequences by the results of Jerrum and Sinclair \cite{Jerrum1990}).

\subsection{From approximate sampling from \texorpdfstring{$\Gmul$}{} to approximating \texorpdfstring{$|\Gmul|$}{}}
\label{app:sm_to_cm}
Using a standard reduction technique, see, e.g., Chapter 12 in \cite{Jerrum1996}, we can  turn our FPAUS into an FPRAS for counting the number of graphs with given degree intervals. 

We first show how to express $|\mathcal{G}_m(\bm{\ell},\bm{u})|$ as a \emph{telescoping product}. We write
\begin{equation}\label{eq:telescope}
	|\mathcal{G}_m(\bm{\ell},\bm{u})| = \frac{|\mathcal{G}_m(\bm{\ell},\bm{u})|}{|\mathcal{G}_m(\bm{a}^{1},\bm{u})|}\frac{|\mathcal{G}_m(\bm{a}^{1},\bm{u})|}{|\mathcal{G}_m(\bm{a}^{2},\bm{u})|}\cdots \frac{|\mathcal{G}_m(\bm{a}^{p-1},\bm{u})|}{|\mathcal{G}_m(\bm{u})|}|\mathcal{G}_m(\bm{u})|
\end{equation}
for a sequence of vectors $\bm{\ell} = \bm{a}^0,\bm{a}^1,\bm{a}^2,\dots,\bm{a}^p = \bm{u}$,
that are  recursively defined as follows. We define $\bm{a}^{i+1}$ by choosing the lowest indexed nodes\footnote{\ A number $i \in [n]$ is lower indexed than $j \in [n]$ if $i < j$.} $v$ and $w$ for which $a^i_v < u^i_v$ and $a^i_w < u^i_w$, and then setting
\[
a^{i+1}_j = \left\{\begin{array}{ll}a^{i}_j + 1 & \text{ if } j \in \{v,w\} \,,\\
	a^{i}_j & \text{ otherwise} \,.
\end{array}\right.
\]
It is clear that there is some $p \leq 2\sum_{i}u_i$ such that this procedure gives $\bm{a}^p = \bm{u}$. Also, if $c = \max_i (u_i - \ell_i)$, all intermediate degree intervals $[\bm{a}^i,\bm{u}]$ also have length at most $c$ component-wise, as $u_j-a^i_j\le u_j-\ell_j\le c$.  Finally, note that we have
\begin{equation}\label{eq:componentwise}
	\bm{a}^0 < \bm{a}^1 < \dots < \bm{a}^{p-1} < \bm{a}^p,
\end{equation}
where for two sequences $\bm{a}$ and $\bm{b}$, we write $\bm{a} < \bm{b}$ if $\bm{a} \le \bm{b}$ and $a_i < b_i$ for at least one $i \in \{1,\dots,n\}$.

In order to approximate the size of $\mathcal{G}_m(\bm{\ell},\bm{u})$, it will be sufficient to approximate the ratios in the telescoping product, as well as the last factor $|\mathcal{G}_m(\bm{u})|$. The latter can be approximated by employing, e.g., the approximate counting scheme of Jerrum and Sinclair \cite{Jerrum1990} for \emph{$P$-stable} degree sequences. For approximating the ratios, we need the following two (sufficient) components: (1) the existence of a FPAUS and  (2) the fact that the ratios can be polynomially bounded. This implies that polynomially many  samples are  enough in order to estimate the ratio up to the desired accuracy.

We sketch how to formalize this argument. Using strong stability, i.e., Proposition \ref{prop:strong}, and very similar reasoning as in the proofs of Lemmas \ref{lem:partial1} and \ref{lem:partial2}, it is easy to show that all ratios in \eqref{eq:telescope} are upper bounded by some polynomial $p_2(n)$ (independent of $\bm{\ell}$ and $\bm{u}$). By setting
\[
\nu_i = \frac{|\mathcal{G}_m(\bm{a}^{i},\bm{u})|}{|\mathcal{G}_m(\bm{a}^{i+1},\bm{u})|} \,,
\]
this just means
$
1 \leq \nu_i \leq p_2(n)
$.
Moreover, using \eqref{eq:componentwise}, it follows that, for $i = 0,\dots,p-1$, we have
$
\mathcal{G}_m(\bm{a}^{i+1},\bm{u}) \subseteq \mathcal{G}_m(\bm{a}^{i},\bm{u})
$.
If we define $X_i$ to be the indicator variable of the event that a random sample from $\mathcal{G}_m(\bm{a}^{i},\bm{u})$ is indeed contained in $\mathcal{G}_m(\bm{a}^{i+1},\bm{u})$, then $\nu_i = 1/\mathbb{E}(X_i)$. The high-level idea is now to show that polynomially many samples from the sampler (the switch hinge-flip Markov chain) not only suffice for an accurate estimate for $\nu_i$ but, crucially, they suffice for an accurate estimate of the product
\[
\prod_{i = 0}^{p-1} \nu_i = \frac{|\mathcal{G}_m(\bm{\ell},\bm{u})|}{|\mathcal{G}_m(\bm{u})|}
\]
up to a factor $(1 \pm \epsilon/3)$. This can be done by standard arguments, e.g., see Chapter 12 of \cite{Jerrum1996} or Chapter 3.2 of \cite{Jerrum2003book}. 
Finally, as mentioned above, we may use the approximate counter from \cite{Jerrum1990} for approximating $|\mathcal{G}_m(\bm{u})|$ up to a factor $(1 \pm \epsilon/3)$. This then implies that we can also approximate $|\mathcal{G}_m(\bm{\ell},\bm{u})|$ up to a factor $(1 \pm  \epsilon)$.

\subsection{From approximating \texorpdfstring{$|\Gmul|$}{} to approximating \texorpdfstring{$|\Gul|$}{}}
\label{app:cm_to_c}
In order to provide an FPRAS for approximating $|\Gul|$, it suffices to give an FPRAS for approximating $|\Gmul|$ for every $\frac{1}{2}\sum_i \ell_i \leq m \leq \frac{1}{2}\sum_i u_i$. Recall that $\Gmul$ is the set of graphs with degree intervals $[\bm{\ell},\bm{u}]$ and for which the total number of edges is equal to $m$.

\begin{lemma}
	Suppose there is an FPRAS for approximating $|\Gmul|$ for every $\frac{1}{2}\sum_i \ell_i \leq m \leq \frac{1}{2}\sum_i u_i$. Then there is an FPRAS for approximating $|\Gul|$.
\end{lemma}
\begin{proof}
	We write $a = \frac{1}{2}\sum_i \ell_i$ and $b = \frac{1}{2}\sum_i u_i$. Note that there are at most $b - a \leq n^2$ possible choices for $m$, and that
	\[
	|\Gul| = \sum_{m = a}^{b} |\Gmul| \,.
	\]
	Now, for every $m$ use the given FPRAS for approximating $|\Gmul|$ with $\delta' = \delta/n^2$. It outputs a number $c_m$ satisfying $(1-\epsilon)|\Gmul|  \leq c_m \leq (1+\epsilon)|\Gmul|$ with probability at least $1 - \delta/n^2$. Then $c = \sum_m c_m$ satisfies $(1-\epsilon)|\Gul|  \leq c \leq (1+\epsilon)|\Gul|$ with probability at least
	\[
	(1 - \delta/n^2)^{b - a} \geq (1 - \delta/n^2)^{n^2} \geq 1- \delta \,.
	\]
	This completes the proof.
\end{proof}

\subsection{From approximating \texorpdfstring{$|\Gul|$}{} to approximate sampling from \texorpdfstring{$\Gul$}{}}
\label{app:c_to_s}

Again, using a reduction inspired by a similar one for self-reducible problems, see, e.g., \cite{JerrumVV86}, we can  turn our FPRAS for computing $|\Gul|$ into an FPAUS for sampling from $\Gul$. Note that if $\bm{\ell} = \bm{u}$, we can simply use the approximate sampler from Jerrum and Sinclair \cite{Jerrum1990} when $\bm{\ell}$ is near-regular.

As long as $\bm{\ell} \neq \bm{u}$ there is some $i$ such that $\ell_i < u_i$. We can partition the set $\Gul$ based on whether or not the degree of a graphical realization $G$ with $\bm{\ell} \leq \bm{d}(G) \leq \bm{u}$ satisfies $d_i = \ell_i$ or $d_i \geq \ell_i + 1$. For a given vector $\bm{z} = (z_1,\dots,z_n) \in \R^n$ and $z_i' \in \R$ we write $(z_i',\bm{z}_{-i}) = (z_1,\dots,z_{i-1},z_i',z_{i+1},\dots,z_n)$. We then have $\Gul$ as the disjoint union
\[
\Gul = \mathcal{G}(\bm{\ell},(\ell_i,\bm{u}_{-i})) \cup \mathcal{G}((\ell_i + 1,\bm{\ell}_{-i}),\bm{u}) \,.
\]
Using the approximation scheme, we can approximate the marginal probabilities 
\[
\frac{|\mathcal{G}(\bm{\ell},(\ell_i,\bm{u}_{-i}))|}{|\Gul|} \ \ \text{\ \  and\ \ }  \ \ \frac{|\mathcal{G}((\ell_i + 1,\bm{\ell}_{-i}),\bm{u})|}{|\Gul|}
\]
up to a sufficient accuracy, and then sample one of the sets $\mathcal{G}(\bm{\ell},(\ell_i,\bm{u}_{-i}))$ or $\mathcal{G}((\ell_i + 1,\bm{\ell}_{-i}),\bm{u})$ according to these---sufficiently accurate---marginals.

We then repeat this step until the lower and upper bound defining the intervals are equal. Note that this step is only carried out a polynomial number of times. After this we have, roughly speaking, sampled a degree sequence $\myd$ with $\bm{\ell} \leq \bm{d} \leq \bm{u}$ according to the (approximately) correct marginal probability. After this we can use the approximate sampler from \cite{Jerrum1990} to sample from $\mathcal{G}(\bm{\myd})$ (or, e.g., simply the switch Markov chain).

\end{document}